\documentclass[a4paper,english,cleveref,autoref,thm-restate]{lipics-v2021}

\usepackage{bbm}
\usepackage{tikz}
\usetikzlibrary{arrows}
\usetikzlibrary{arrows.meta}
\usetikzlibrary{decorations.pathmorphing}
\usetikzlibrary{decorations.markings}
\usetikzlibrary{calc}

\tikzset{
  cir/.style = {circle,draw,fill,inner sep=.7pt},
  circ/.style = {circle,draw,fill,inner sep=1.3pt},
  circg/.style = {circle,draw=lightgray,fill=lightgray,inner sep=1.3pt},
  circr/.style = {circle,draw=red,fill=red,inner sep=1.3pt},
  invisible/.style = {circle,draw=none,inner sep=0pt,font=\tiny},
  nonedge/.style={decorate,decoration={snake,amplitude=.3mm,segment length=1mm},draw}
}

\newcommand{\vol}{\mathrm{vol}}
\newcommand{\rk}{\mathrm{rk}}
\newcommand{\tw}{\mathrm{tw}}
\newcommand{\tin}{\mathrm{tree}\textnormal{-}\alpha}

\title{Polynomial-Time Approximation Schemes for Independent Packing Problems on Fractionally Tree-Independence-Number-Fragile Graphs} 

\titlerunning{PTASes for Independent Packing Problems on Fractionally $\tin$-Fragile Graphs} 

\author{Esther {Galby}}{Hamburg University of Technology, Institute for Algorithms and Complexity,
Hamburg, Germany}{esther.galby@tuhh.de}{}{}

\author{Andrea {Munaro}}{Department of Mathematical, Physical and Computer Sciences, University of Parma, Parma, Italy}{andrea.munaro@unipr.it}{https://orcid.org/0000-0003-1509-8832}{}

\author{Shizhou {Yang}}{School of Mathematics and Physics, Queen’s University Belfast, Belfast, UK}{syang22@qub.ac.uk}{}{}

\authorrunning{E. Galby, A. Munaro, S. Yang} 

\Copyright{Esther Galby and Andrea Munaro and Shizhou Yang} 



\category{} 

\acknowledgements{We thank the anonymous reviewers for valuable comments.}

\nolinenumbers 

\hideLIPIcs

\EventEditors{Erin W. Chambers and Joachim Gudmundsson}
\EventNoEds{2}
\EventLongTitle{39th International Symposium on Computational Geometry
(SoCG 2023)}
\EventShortTitle{SoCG 2023}
\EventAcronym{SoCG}
\EventYear{2023}
\EventDate{June 12--15, 2023}
\EventLocation{Dallas, Texas, USA}
\EventLogo{}
\SeriesVolume{258}
\ArticleNo{XX} 

\begin{document}

\maketitle

\begin{abstract}
We investigate a relaxation of the notion of treewidth-fragility, namely tree-independence-number-fragility. In particular, we obtain polynomial-time approximation schemes for independent packing problems on fractionally tree-independence-number-fragile graph classes. Our approach unifies and extends several known polynomial-time approximation schemes on seemingly unrelated graph classes, such as classes of intersection graphs of fat objects in a fixed dimension or proper minor-closed classes. We also study the related notion of layered tree-independence number, a relaxation of layered treewidth.
\end{abstract}

\section{Introduction}
\label{sec:introA}

Many optimization problems involving collections of geometric objects in the $d$-dimensional space are known to admit a polynomial-time approximation scheme (PTAS). Arguably the earliest example of such behavior is the problem of finding the maximum number of pairwise non-intersecting disks or squares in a collection of unit disks or unit squares, respectively \cite{HM85}. Such subcollection is usually called an \textit{independent packing}. This result was later extended to collections of arbitrary disks and squares and, more generally, fat objects \cite{Cha03,EJS05}. The reason for the abundance of approximation schemes for geometric problems is that shifting and layering techniques can be used to reduce the problem to small subproblems that can be solved by dynamic programming. In fact, the same phenomenon occurs for graph problems, as evidenced by the seminal work of Baker \cite{Bak94} on approximation schemes for local problems, such as \textsc{Independent Set}, on planar graphs and its generalizations first to apex-minor-free graphs \cite{Epp00} and further to graphs embeddable on a surface of bounded genus with a bounded number of crossings per edge \cite{GB07}. The notion of intersection graph allows to jump from the geometric world to the graph-theoretic one. Given a collection $\mathcal{O}$ of geometric objects in $\mathbb{R}^d$, we can consider its \textit{intersection graph}, the graph whose vertices are the objects in $\mathcal{O}$ and where two vertices $O_i, O_j \in \mathcal{O}$ are adjacent if and only if $O_i \cap O_j \neq \varnothing$. An independent packing in $\mathcal{O}$ is then nothing but an independent set in the corresponding intersection graph. Notice that intersection graphs of unit disks or squares are not minor-closed, as they contain arbitrarily large cliques. Our motivating question is the following:
\begin{center}
\textit{Is there any underlying graph-theoretical reason for the existence of the seemingly unrelated PTASes for \textsc{Independent Set} mentioned above?} 
\end{center}
We provide a positive answer to this question that also allows us to further generalize to a family of independent packing problems. We also remark that the similar question of whether there is a general notion under which PTASes using Baker's technique can be obtained was already asked in \cite{GB07}. 

Baker's layering technique relies on a form of decomposition theorem for planar graphs that can be roughly summarized as follows. Given a planar graph $G$ and $k \in \mathbb{N}$, the vertex set of $G$ can be partitioned into $k+1$ possibly empty sets in such a way that deleting any part induces a graph of treewidth at most $O(k)$ in $G$. Moreover, such a partition together with tree decompositions of width at most $O(k)$ of the respective graphs can be found in polynomial time. A statement of this form is typically referred to as a \textit{Vertex Decomposition Theorem} (VDT) \cite{PSZ19}. VDTs are known to exist in planar graphs \cite{Bak94}, graphs of bounded-genus and apex-minor-free graphs \cite{Epp00}, and $H$-minor-free graphs \cite{DHK05,DDO04}. However, their existence is in general something too strong to ask for, as is the case of intersection graphs of unit disks or squares and hence fat objects in general. There are then two natural ways in which one can try to relax the notion of VDT. First, we can consider an approximate partition of the vertex set, where a vertex can belong to some constant number of sets. Second, we can look for a width parameter less restrictive than treewidth.

Dvo\v{r}\'{a}k \cite{Dvo16} pursued the first direction and introduced the notion of \textit{efficient fractional treewidth-fragility}. We state here an equivalent formulation from \cite{DL21}. A class of graphs $\mathcal{G}$ is \textit{efficiently fractionally treewidth-fragile} if there exists a function $f\colon \mathbb{N} \rightarrow \mathbb{N}$ and an algorithm that, for every $k \in \mathbb{N}$ and $G \in \mathcal{G}$, returns in time $\mathsf{poly}(|V(G)|)$ a collection of subsets $X_1, X_2,\ldots, X_m \subseteq V(G)$ such that each vertex of $G$ belongs to at most $m/k$ of the subsets and moreover, for $i = 1,\ldots, m$, the algorithm also returns a tree decomposition of $G - X_i$ of width at most $f(k)$. Several graph classes are known to be efficiently fractionally treewidth-fragile. In fact, a hereditary class $\mathcal{G}$ is efficiently fractionally treewidth-fragile in each of the following cases (see, e.g., \cite{DL21}): $\mathcal{G}$ has sublinear separators and bounded maximum degree, $\mathcal{G}$ is proper minor-closed, or $\mathcal{G}$ consists of intersection graphs of convex objects with bounded aspect ratio in $\mathbb{R}^d$ (for fixed $d$) and the graphs in $\mathcal{G}$ have bounded clique number. Dvo\v{r}\'{a}k \cite{Dvo16} showed that \textsc{Independent Set} admits a PTAS on every efficiently fractionally treewidth-fragile graph class. This result was later extended \cite{Dvo22,DL21} to a framework of maximization problems including, for example, \textsc{Max Weight Distance-$d$ Independent Set}, \textsc{Max Weight Induced Forest} and \textsc{Max Weight Induced Matching}. However, the notion of fractional treewidth-fragility falls short of capturing classes such as unit disk graphs, as it implies the existence of sublinear separators \cite{Dvo16}.  

One can then try to pursue the second direction mentioned above and further relax the notion of efficient fractional fragility by considering width parameters \textit{more powerful} than treewidth (i.e., bounded on a larger class of graphs) and algorithmically useful. A natural candidate is the recently introduced \textit{tree-independence number} \cite{DaMS22}, a width parameter defined in terms of tree decompositions which is more powerful than treewidth (see \Cref{sec:layeredA}). Several algorithmic applications of boundedness of tree-independence number have been provided, most notably polynomial-time solvability of \textsc{Max Weight Independent Packing} \cite{DaMS22} (see \Cref{sec:ptasesA} for the definition), a common generalization of \textsc{Max Weight Independent Set} and \textsc{Max Weight Induced Matching}, and of its distance-$d$ version, for $d$ even \cite{MR22}. Investigating the notion of efficient fractional tree-independence-number-fragility ($\tin$-fragility for short) was recently suggested in a talk by Dvo\v{r}\'{a}k \cite{GWP22}, where it was stated that, using an argument from \cite{DGLTT22}, it is possible to show that intersection graphs of balls and cubes in $\mathbb{R}^d$ are fractionally $\tin$-fragile.   

A successful notion related to fractional treewidth-fragility is the layered treewidth of a graph \cite{DMW17}. Despite currently lacking any direct algorithmic application, it proved useful especially in the context of coloring problems (we refer to \cite{DEMWW22} for additional references). We just mention that classes of bounded layered treewidth include planar graphs and, more generally, apex-minor-free graphs and graphs embeddable on a surface of bounded genus with a bounded number of crossings per edge, amongst others \cite{DEW17}. It can be shown that bounded layered treewidth implies fractional treewidth-fragility (see \Cref{sec:fragilityA}). Layered treewidth is also related to local treewidth, a notion first introduced by Eppstein \cite{Epp00}, and in fact, on proper minor-closed classes, having bounded layered treewidth coincides with having bounded local treewidth (see, e.g., \cite{DEW17}). 

\subsection{Our results}

In this paper, we investigate the notion of efficient fractional $\tin$-fragility and show that it answers our motivating question in the positive, thus allowing to unify and extend several known results. Our main result can be summarized as follows and will be proved in \Cref{sec:fragilityA} and \Cref{sec:ptasesA}.

\begin{theorem}\label{mainA} \textsc{Max Weight Independent Packing} admits a PTAS on every efficiently fractionally $\tin$-fragile class. Moreover, the class of intersection graphs of fat objects in $\mathbb{R}^d$, for fixed $d$, is efficiently fractionally $\tin$-fragile\footnote{Here we use a definition of fatness slightly generalizing that of Chan \cite{Cha03} (see \Cref{fatA}).}. 
\end{theorem}

The message of \Cref{mainA} is that a doubly-relaxed version of a VDT suffices for algorithmic applications and is general enough to hold for several interesting graph classes. \Cref{mainA} cannot be improved to guarantee an EPTAS, unless $\mathsf{FPT} = \mathsf{W}[1]$. Indeed, Marx \cite{Mar05} showed that \textsc{Independent Set} remains $\mathsf{W}[1]$-complete on intersection graphs of unit disks and unit squares. The natural trade-off in extending the tractable families with respect to approximation is that fewer problems will admit a PTAS. In our case this is exemplified by the minimization problem \textsc{Feedback Vertex Set}, which admits no PTAS, unless $\mathsf{P} = \mathsf{NP}$, on unit ball graphs in $\mathbb{R}^3$ \cite{FLS12} but admits an EPTAS on disk graphs in $\mathbb{R}^2$ \cite{LPS23}.  

In \Cref{sec:fragilityA}, we also show that fractionally $\tin$-fragile classes have bounded biclique number, where the \textit{biclique number} of a graph $G$ is the maximum $n \in \mathbb{N}$ such that the complete bipartite graph $K_{n,n}$ is an induced subgraph of $G$. This shows in particular that, unsurprisingly, the notion of fractional $\tin$-fragility falls short of capturing intersection graphs of rectangles in the plane. Whether \textsc{Independent Set} admits a PTAS on these graphs remains one of the major open problems in the area (see, e.g., \cite{GKM22}). We also show that the absence of large bicliques is not sufficient for guaranteeing fractional $\tin$-fragility: $n$-dimensional grids of width $n$ are $K_{2,3}$-free but not fractionally $\tin$-fragile. 

We begin our study of fractional $\tin$-fragility by introducing, in \Cref{sec:layeredA}, a subclass of fractionally $\tin$-fragile graphs, namely the class of graphs with bounded layered tree-independence number. We obtain the notion of \textit{layered tree-independence number} by relaxing the successful notion of layered treewidth and show that, besides graphs of bounded layered treewidth, classes of intersection graphs of unit disks in $\mathbb{R}^2$ and of paths with bounded horizontal part on a grid have bounded layered tree-independence number. Moreover, we observe that, for minor-closed classes, having bounded layered tree-independence number is equivalent to having bounded layered treewidth, thus extending a characterization of bounded layered treewidth stated in \cite{DEW17}. 

We then consider the behavior of layered tree-independence number with respect to graph powers. We show that odd powers of graphs of bounded layered tree-independence number have bounded layered tree-independence number and that this does not extend to even powers. Combined with \Cref{mainA}, this gives the following result which applies, for example, to unit disk graphs and cannot be extended to odd $d \in \mathbb{N}$ (see \Cref{sec:distanceA}).

\begin{theorem}\label{distancelayeredA} For a fixed positive even integer $d$, the distance-$d$ version of \textsc{Max Weight Independent Packing} admits a PTAS on every class of bounded layered tree-independence number, provided that a tree decomposition and a layering witnessing small layered tree-independence number can be computed efficiently.
\end{theorem}

Finally, we show that the approach to PTASes through tree-independence number is competitive in terms of running time for some classes of intersection graphs. Specifically, in \Cref{sec:improvedA}, we obtain PTASes for \textsc{Max Weight Independent Set} for intersection graphs of families of unit disks, unit-height rectangles, and paths with bounded horizontal part on a grid, which improve results from \cite{Mat00,Cha04,BBCGP20}, respectively, mentioned in the next section.

We believe that the notion of fractional $\tin$-fragility can find further applications in the design of PTASes. In fact, it would be interesting to obtain an algorithmic meta-theorem similar to those for fractionally treewidth-fragile classes \cite{DL21,Dvo22} and classes of bounded tree-independence number \cite{MR22}. Although our interest is in approximation schemes, we notice en passant that the observations from \Cref{sec:layeredA} lead to a subexponential-time algorithm for the distance-$d$ version of \textsc{Max Weight Independent Packing}, for $d$ even, on unit disk graphs. We finally remark that all our PTASes for intersection graphs of geometric objects are not robust i.e., they require a geometric realization to be part of the input. 

\subsection{Other related work}

\hspace{\parindent}\textbf{Disk graphs.} Very recently, Lokshtanov et al. \cite{LPS23} established a framework for designing EPTASes for a broad class of minimization problems (specifically, vertex-deletion problems) on disk graphs including, among others, \textsc{Feedback Vertex Set} and $d$-\textsc{Bounded Degree Vertex Deletion}. Previous sporadic PTASes on this class were known only for \textsc{Vertex Cover} \cite{EJS05,vLee06}, \textsc{Dominating Set} \cite{GB10}, \textsc{Independent Set} \cite{Cha03,EJS05} and \textsc{Max Clique} \cite{BBB21}. \Cref{mainA} adds several maximization problems to this list (see \Cref{sec:ptasesA}).

\textbf{Unit disk graphs.} Unit disk graphs are arguably one of the most well-studied graph classes in computational geometry, as they naturally model several real-world problems. Great attention has been devoted to approximation algorithms for \textsc{Max Weight Independent Set} on this class (see, e.g., \cite{HMR98,NHK04,vLee05}). To the best of our knowledge, the fastest known PTAS is a $(1-1/k)$-approximation algorithm with running time $O(k n^{4\lceil \frac{2(k-1)}{\sqrt{3}}\rceil})$ \cite{Mat00}. We also remark that a special type of Decomposition Theorem was recently shown to hold for the class of unit disk graphs. A Contraction Decomposition Theorem (CDT) is a statement of the following form: given a graph $G$, for any $p \in \mathbb{N}$, one can partition the edge set of $G$ into $E_1,\ldots, E_p$ such that contracting the edges in each $E_i$ in $G$ yields a graph of treewidth at most $f(p)$, for some function $f\colon \mathbb{N} \rightarrow \mathbb{N}$. CDTs are useful in designing efficient approximation and parameterized algorithms and are known to hold for classes such as graphs of bounded genus \cite{DHM10} and unit disk graphs \cite{BLLSX22}. Since these classes are efficiently fractionally $\tin$-fragile, our results can be seen as providing a different type of relaxed decomposition theorems for them.

\textbf{Intersection graphs of unit-height rectangles.} As observed by Agarwal et al. \cite{AKS98}, this class of graphs arises naturally as a model for the problem of labeling maps with labels of the same font size. Improving on \cite{HM85}, they obtained a $(1 - 1/k)$-approximation algorithm for \textsc{Max Weight Independent Set} on this class with running time $O(n^{2k-1})$. Chan \cite{Cha04} provided a $(1-1/k)$-approximation algorithm with running time $O(n^k)$.

\textbf{Intersection graphs of paths on a grid.} Asinowski et al. \cite{asinowski} introduced the class of \textit{Vertex intersection graphs of Paths on a Grid} (\textit{VPG graphs} for short). A graph $G$ is a \textit{VPG graph} if there exists a collection $\mathcal{P}$ of paths on a grid $\mathcal{G}$ such that $\mathcal{P}$ is in one-to-one correspondence with $V(G)$ and two vertices are adjacent in $G$ if and only if the corresponding paths intersect. It is not difficult to see that this class coincides with the well-known class of string graphs. If every path in $\mathcal{P}$ has at most $k$ \textit{bends} i.e., $90$ degrees turns at a grid-point, the graph is a \textit{$B_k$-VPG graph}. Golumbic et al. \cite{GLS09} introduced the class of \textit{Edge intersection graphs of Paths on a Grid} (\textit{EPG graphs} for short) which is defined similarly to VPG, except that two vertices are adjacent if and only if the corresponding paths share a grid-edge. It turns out that every graph is EPG \cite{GLS09} and $B_{k}$-EPG graphs have been defined similarly to $B_{k}$-VPG graphs. Approximation algorithms for \textsc{Independent Set} on VPG and EPG graphs have been deeply investigated, especially when the number of bends is a small constant (see, e.g., \cite{BD17,FP11,LMS15,Meh17}). It is an open problem whether \textsc{Independent Set} admits a PTAS on $B_1$-VPG graphs \cite{Meh17}. Concerning EPG graphs, Bessy et al. \cite{BBCGP20} showed that the problem admits no PTAS on $B_{1}$-EPG graphs, unless $\mathsf{P} = \mathsf{NP}$, even if each path has its vertical segment or its horizontal segment of length at most $1$. On the other hand, they provided a PTAS for \textsc{Independent Set} on $B_1$-EPG graphs where the length of the horizontal part\footnote{The \textit{horizontal part} of a path is the interval corresponding to the projection of the path onto the $x$-axis.} of each path is at most a constant $c$ with running time $O^{*}(n^{\frac{3c}{\varepsilon}})$.


\section{Preliminaries}
\label{sec:prelimA}

We consider only finite simple graphs. If $G'$ is a subgraph of $G$ and $G'$ contains all the edges of $G$ with both endpoints in $V(G')$, then $G'$ is an \textit{induced subgraph} of $G$ and we write $G' = G[V(G')]$. For a vertex $v \in V(G)$ and $r \in \mathbb{N}$, the $r$-\textit{closed neighborhood} $N^{r}_{G}[v]$ is the set of vertices at distance at most $r$ from $v$ in $G$. The \textit{degree} $d_{G}(v)$ of a vertex $v \in V(G)$ is the number of edges incident to $v$ in $G$. The \textit{maximum degree} $\Delta(G)$ of $G$ is the quantity $\max\left\{d_{G}(v): v \in V\right\}$. Given a graph $G = (V, E)$ and $V' \subseteq V$, the operation of \textit{deleting the set of vertices $V'$} from $G$ results in the graph $G - V' = G[V\setminus V']$. A graph is \textit{$Z$-free} if it does not contain induced subgraphs isomorphic to graphs in a set $Z$. The complete bipartite graph with parts of sizes $r$ and $s$ is denoted by $K_{r,s}$. An \textit{independent set} of a graph is a set of pairwise non-adjacent vertices.  The maximum size of an independent set of $G$ is denoted by $\alpha(G)$. A \textit{clique} of a graph is a set of pairwise adjacent vertices. A \textit{matching} of a graph is a set of pairwise non-incident edges. An \textit{induced matching} in a graph is a matching $M$ such that no two vertices belonging to different edges in $M$ are adjacent in the graph.

\textbf{Intersection graphs of unit disks and rectangles.} We now explain how the geometric realizations of these intersection graphs are encoded. A collection of unit disks with a common radius $c \in \mathbb{R}$ is encoded by a collection of points in $\mathbb{R}^2$ representing the centers of the disks. Unless otherwise stated, when we refer to a rectangle we mean an axis-aligned closed rectangle in $\mathbb{R}^2$. As is typically done for intersection graphs of rectangles, we assume that the vertices of the rectangles are on an integer grid $\mathcal{G}$ and each rectangle is encoded by the coordinates of its vertices. Given an intersection graph $G$ of a family $\mathcal{R}$ of rectangles, a \textit{grid representation} of $G$ is a pair $(\mathcal{G}, \mathcal{R})$ as above.   
 
\textbf{VPG and EPG graphs.} Given a rectangular grid $\mathcal{G}$, its horizontal lines are referred to as \textit{rows} and its vertical lines as \textit{columns}. For a VPG (EPG) graph $G$, the pair $\mathcal{R} = (\mathcal{G},\mathcal{P})$ is a \textit{VPG representation} (\textit{EPG representation}) of $G$. More generally, a \emph{grid representation} of a graph $G$ is a triple $\mathcal{R} = (\mathcal{G},\mathcal{P},x)$ where $x \in \{e,v\}$, such that $(\mathcal{G},\mathcal{P})$ is an EPG representation of $G$ if $x = e$, and $(\mathcal{G},\mathcal{P})$ is a VPG representation of $G$ if $x = v$. Note that, irrespective of whether $x=e$ (that is, $G$ is an EPG graph) or $x=v$ (that is, $G$ is a VPG graph), if two vertices $u,v \in V(G)$ are adjacent in $G$ then $P_u$ and $P_v$ share at least one grid-point. A \textit{bend-point} of a path $P\in\mathcal{P}$ is a grid-point corresponding to a bend of $P$ and a \textit{segment} of $P$ is either a vertical or horizontal line segment in the polygonal curve constituting $P$. Paths in $\mathcal{P}$ are encoded as follows. For each $P \in \mathcal{P}$, we have one sequence $s(P)$ of points in $\mathbb{R}^2$: $s(P) = (x_1,y_1), (x_2,y_2), \dots, (x_{\ell_P},y_{\ell_P})$ consists of the endpoints $(x_1,y_1)$ and $(x_{\ell_P},y_{\ell_P})$ of $P$ and all the bend-points of $P$ in their order of appearance when traversing $P$ from $(x_1,y_1)$ to $(x_{\ell_P},y_{\ell_P})$. If each path in $\mathcal{P}$ has a number of bends polynomial in $|V(G)|$, then the size of this data structure is polynomial in $|V(G)|$. Given $s(P)$, we can easily determine the horizontal part $h(P)$ of the path $P$ (i.e., the projection of $P$ onto the $x$-axis). Note that our results for VPG and EPG graphs (\Cref{layeredVPGA,indepPTASA}), although stated for constant number of bends, still hold for polynomial (in $|V(G)|$) number of bends, with a worse polynomial running time.  

\textbf{PTAS.} A PTAS for a maximization problem is an algorithm which takes an instance $I$ of the problem and a parameter $\varepsilon > 0$ and produces a solution within a factor $1 - \varepsilon$ of the optimal in time $n^{O(f(1/\varepsilon))}$. A PTAS with running time $f(1/\varepsilon)\cdot n^{O(1)}$ is called an efficient PTAS (EPTAS for short).


\section{Layered and local tree-independence number}\label{sec:layeredA}

The key definitions of this section are those of tree-independence number and layering, which we now recall. A \textit{tree decomposition} of a graph $G$ is a pair $\mathcal{T} = (T, \{X_t\}_{t\in V(T)})$, where $T$ is a tree whose every node $t$ is assigned a vertex subset $X_t \subseteq V(G)$, called a \textit{bag}, such that the following conditions are satisfied: 
\begin{itemize}
 \setlength{\itemindent}{1em}
\item[(T1)] Every vertex of $G$ belongs to at least one bag; 
\item[(T2)] For every $uv \in E(G)$, there exists a bag containing both $u$ and $v$; 
\item[(T3)] For every $u \in V(G)$, the subgraph $T_u$ of $T$ induced by $\{t \in V(T) : u \in X_t\}$ is connected. 
\end{itemize}
The \textit{width} of $\mathcal{T} = (T, \{X_t\}_{t\in V(T)})$ is the maximum value of $|X_t| - 1$ over all $t \in V(T)$. The \textit{treewidth} of a graph $G$, denoted $\tw(G)$, is the minimum width of a tree decomposition of $G$. The \textit{independence number} of $\mathcal{T}$, denoted $\alpha(\mathcal{T}$), is the quantity $\max_{t\in V(T)} \alpha(G[X_t])$. The \textit{tree-independence number} of a graph $G$, denoted $\tin(G)$, is the minimum independence number of a tree decomposition of $G$. Clearly, $\tin(G) \leq \tw(G)+1$, for any graph $G$. On the other hand, tree-independence number is a width parameter more powerful than treewidth, as there exist classes with bounded tree-independence number and unbounded treewidth (for example, chordal graphs have tree-independence number $1$ \cite{DaMS22}).

A \textit{layering} of a graph $G$ is a partition $(V_0, V_1,\ldots, V_t)$ of $V(G)$ such that, for every edge $vw \in E(G)$, if $v \in V_i$ and $w \in V_j$, then $|i - j| \leq 1$. Each set $V_i$ is a \textit{layer}. 
The \textit{layered width} of a tree decomposition $\mathcal{T} = (T, \{X_t\}_{t\in V(T)})$ of a graph $G$ is the minimum integer $\ell$ such that, for some layering $(V_0, V_1,\ldots)$ of $G$, and for each bag $X_t$ and layer $V_i$, we have $|X_t \cap V_i| \leq \ell$. The \textit{layered treewidth} of a graph $G$ is the minimum layered width of a tree decomposition of $G$. Layerings with one layer show that the layered treewidth of $G$ is at most $\tw(G)+1$. We now introduce the analogue of layered treewidth for the width parameter tree-independence number.

\begin{definition} The \textit{layered independence number} of a tree decomposition $\mathcal{T} = (T, \{X_t\}_{t\in V(T)})$ of a graph $G$ is the minimum integer $\ell$ such that, for some layering $(V_0, V_1,\ldots)$ of $G$, and for each bag $X_t$ and layer $V_i$, we have $\alpha(G[X_t \cap V_i]) \leq \ell$. The \textit{layered tree-independence number} of a graph $G$ is the minimum layered independence number of a tree decomposition of $G$. 
\end{definition}

Layerings with one layer show that the layered tree-independence number of $G$ is at most $\tin(G)$. Moreover, the layered tree-independence number of a graph is clearly at most its layered treewidth. The proof of \cite[Lemma~10]{DMW17} shows, mutatis mutandis, that graphs of bounded layered tree-independence number have $O(\sqrt{n})$ tree-independence number:  

\begin{lemma}\label{sqrttreealphaA} Every $n$-vertex graph with layered tree-independence number $k$ has tree-independence number at most $2\sqrt{kn}$.
\end{lemma}

Given a width parameter $p$, a graph class $\mathcal{G}$ has \textit{bounded local $p$} if there is a function $f\colon \mathbb{N} \rightarrow \mathbb{N}$ such that for every integer $r \in \mathbb{N}$, graph $G \in \mathcal{G}$, and vertex $v \in V(G)$, the subgraph $G[N^{r}[v]]$ has $p$-width at most $f(r)$. In \cite{DMW17}, it is shown that if every graph in a class $\mathcal{G}$ has layered treewidth at most $\ell$, then $\mathcal{G}$ has bounded local treewidth with $f(r) = \ell(2r + 1) - 1$.

\begin{lemma}\label{blayeredblocalA} If every graph in a class $\mathcal{G}$ has layered tree-independence number at most $\ell$, then $\mathcal{G}$ has bounded local tree-independence number with $f(r) = \ell(2r + 1)$.
\end{lemma}

\begin{proof} Let $r \in \mathbb{N}$, $G \in \mathcal{G}$ and $v \in V(G)$. Let $G' = G[N^{r}[v]]$. By assumption, $G$ has a tree decomposition $\mathcal{T} = (T, \{X_t\}_{t\in V(T)})$ of layered independence number $\ell$ with respect to some layering $(V_0, V_1,\ldots)$. Suppose that $v \in V_i$. Then, $V(G') \subseteq V_{i-r} \cup \cdots \cup V_{i+r}$ and so, for each bag $X_t$, we have that $\alpha(G'[X_t]) \leq \sum_{j=-r}^{r}\alpha(G[X_t \cap V_{i-j}]) \leq \ell(2r + 1)$. This implies that $\tin(G') \leq \ell(2r + 1)$.  
\end{proof}

\begin{corollary}\label{layeredKnnA} The layered tree-independence number of $K_{n,n}$ is at least $n/5$.
\end{corollary}

\begin{proof} Suppose, to the contrary, that the layered tree-independence number of $K_{n,n}$ is less than $n/5$. Since the diameter of $K_{n,n}$ is $2$, \Cref{blayeredblocalA} implies that $\tin(K_{n,n}) < n$, contradicting the fact that $\tin(K_{n,n}) = n$ \cite{DaMS22}.  
\end{proof}

\begin{figure}[h!]
\centering
\begin{subfigure}{.45\linewidth}
\centering
\begin{tikzpicture}[scale=.5,rotate=90]

\draw[very thick] (0,0) -- (0,1) -- (1,1) -- (1,2) -- (2,2) -- (2,3) -- (3,3) -- (3,4);
\draw[very thick] (0,2) -- (0,3) -- (1,3) -- (1,4) -- (2,4) -- (2,5) -- (3,5) -- (3,6);
\draw[very thick] (0,4) -- (0,5) -- (1,5) -- (1,6) -- (2,6) -- (2,7) -- (3,7) -- (3,8);
\draw[very thick] (0,6) -- (0,7) -- (1,7) -- (1,8) -- (2,8) -- (2,9) -- (3,9) -- (3,10);

\draw[very thick] (.15,0) -- (.15,7);
\draw[very thick] (1.15,1) -- (1.15,8);
\draw[very thick] (2.15,2) -- (2.15,9);
\draw[very thick] (3.15,3) -- (3.15,10);

\end{tikzpicture}
\end{subfigure}
\begin{subfigure}{.45\linewidth}
\centering
\begin{tikzpicture}[scale=.4]

\draw[thick] (0,0) rectangle (7,1);
\draw[thick] (0,2) rectangle (7,3);
\draw[thick] (0,4) rectangle (7,5);
\draw[thick] (0,6) rectangle (7,7);

\draw[thick] (.15,.15) rectangle (1.05,6.85);
\draw[thick] (2,.15) rectangle (3,6.85);
\draw[thick] (4,.15) rectangle (5,6.85);
\draw[thick] (5.95,.15) rectangle (6.85,6.85);
\end{tikzpicture}
\end{subfigure}
\caption{Examples showing that VPG/EPG graphs and intersection graphs of rectangles have unbounded layered tree-independence number: VPG/EPG representation (left) and representation by intersection of rectangles (right) of $K_{4,4}$.}
\end{figure}
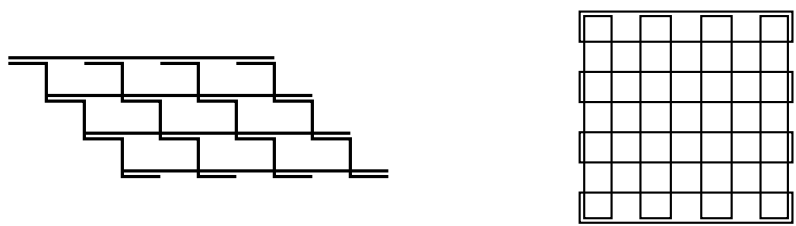

\begin{theorem}\label{equivlayeredA} The following are equivalent for a minor-closed class $\mathcal{G}$:
\begin{enumerate}
\item Some apex\footnote{An apex graph is a graph that can be made planar by deleting a single vertex.} graph is not in $\mathcal{G}$;
\item $\mathcal{G}$ has bounded local tree-independence number;
\item $\mathcal{G}$ has linear local tree-independence number (i.e., $f(r)$ is linear in $r$);
\item $\mathcal{G}$ has bounded layered tree-independence number.
\end{enumerate}
\end{theorem}

\begin{proof} (1) $\Longrightarrow$ (4): By \cite{DMW17}, if $\mathcal{G}$ excludes some apex graph as a minor, then $\mathcal{G}$ has bounded layered treewidth, hence bounded layered tree-independence number as well. 

\noindent (4) $\Longrightarrow$ (3): It follows from \Cref{blayeredblocalA}.

\noindent (3) $\Longrightarrow$ (2): It follows by definition.

\noindent (2) $\Longrightarrow$ (1): Let $G_n$ be the graph obtained from the $n\times n$-grid graph (the Cartesian product of two $n$-vertex paths) by adding a dominating vertex $v_n$. Observe that the class $\{G_n : n\in \mathbb{N}\}$ has unbounded local tree-independence number, as $v_n$ is dominating and the class of grids has unbounded tree-independence number (since it is not $(\tw,\omega)$-bounded, see \cite[Lemma~3.2]{DaMS22}). Hence, if $\mathcal{G}$ contains all apex graphs, then in particular it contains $\{G_n : n\in \mathbb{N}\}$ and so has unbounded local tree-independence number.
\end{proof}

\Cref{equivlayeredA} implies the following result from \cite{DMiS22}.

\begin{corollary} A minor-closed class has bounded tree-independence number if and only if some planar graph is not in the class.
\end{corollary}

\begin{proof} Let $\mathcal{G}$ be a minor-closed class. If $\mathcal{G}$ has bounded tree-independence number then, since the class of walls has unbounded tree-independence number \cite{DMS21a,DaMS22}, $\mathcal{G}$ does not contain a planar graph as a minor. 

Conversely, we claim that, for every planar graph $H$ there is an integer $c$ such that every $H$-minor-free graph $G$ has tree-independence number at most $c$. Let $H^{+}$ be the apex
graph obtained from $H$ by adding a dominating vertex $v$ and let $G^{+}$ be the graph obtained from $G$ by adding a dominating vertex $x$. It is shown in \cite{DMW17} that $G^{+}$ is $H^{+}$-minor-free. By \Cref{equivlayeredA}, $G^{+}$ has layered tree-independence number at most $\ell$, for some fixed integer $\ell$. Since $G^{+}$ has radius $1$, at most three layers are used. Thus $G^{+}$, and hence $G$, have tree-independence number at most $3\ell$.
\end{proof}

For $p \in \mathbb{N}$, the $p$-th power of a graph $G$ is the graph $G^p$ with vertex set $V(G^p) = V(G)$, where $uv \in E(G^p)$ if and only if $u$ and $v$ are at distance at most $p$ in $G$. Bonomo-Braberman and Gonzalez \cite{BG22} showed that fixed powers of bounded treewidth and bounded degree graphs are of bounded treewidth. More specifically, for any graph $G$ and $p \geq 2$, $\tw(G^p) \leq (\tw(G) + 1)(\Delta(G) + 1)^{\lceil\frac{p}{2}\rceil}-1$. It follows from the work of Dujmovi\'c et al. \cite{DMW22} that powers of graphs of bounded layered treewidth and bounded maximum degree have bounded layered treewidth. The upper bound was later improved by Dujmovi\'c et al. \cite{DEMWW22}, who showed that if $G$ has layered treewidth $k$, then $G^p$ has layered treewidth less than $2pk\Delta(G)^{\lfloor\frac{p}{2}\rfloor}$. We show that odd powers of bounded layered tree-independence number graphs have bounded layered tree-independence number and that this result does not extend to even powers. Before doing so, we need a definition and a result from \cite{MR22}. Given a graph $G$ and a family $\mathcal{H} = \{H_j\}_{j\in J}$ of subgraphs of $G$, we denote by $G(\mathcal{H})$ the graph with vertex set $J$, in which two distinct elements $i, j \in J$ are adjacent if and only if $H_i$ and $H_j$ either have a vertex in common or there is an edge in $G$ connecting them.

\begin{lemma}[Milani\v{c} and Rz\c{a}\.{z}ewski \cite{MR22}]\label{sametreeA} Let $G$ be a graph and let $k$ and $d$ be positive integers. For $v \in V(G)$, let $H_v$ be the subgraph of $G$ induced by the vertices at distance at most $d$ from $v$, and let $\mathcal{H} = \{H_v\}_{v\in V(G)}$. Then $G^{k+2d} = G^k(\mathcal{H})$. 
\end{lemma}

\begin{theorem}\label{layeredtreealgoA} Let $G$ be a graph and let $d$ be a positive integer. Given a tree decomposition $\mathcal{T} = (T,\{X_t\}_{t\in V(T)})$ of $G$ and a layering $(V_1,\ldots,V_{m})$ of $G$ such that, for each bag $X_t$ and layer $V_i$, $\alpha(G[X_t \cap V_i]) \leq k$, it is possible to compute in $O(|V(T)| \cdot (|V(G)| + |E(G)|))$ time a tree decomposition $\mathcal{T'} = (T,\{X'_t\}_{t\in V(T)})$ of $G^{1+2d}$ and a layering $(V'_1,\ldots,V'_{\lceil \frac{m}{1+2d} \rceil})$ of $G^{1+2d}$ such that, for each bag $X'_t$ and layer $V'_i$, $\alpha(G^{1+2d}[X'_t \cap V'_i]) \leq (1+4d)k$. In particular, if $G$ has layered tree-independence number $k$, then $G^{1+2d}$ has layered tree-independence number at most $(1+4d)k$.
\end{theorem}

\begin{proof} Let $\mathcal{T} = (T,\{X_t\}_{t\in V(T)})$ and $(V_1,\ldots,V_{m})$ be the given tree decomposition and layering of $G$, respectively. For each vertex $u \in V(G)$, let $l(u)$ be the unique index $i$ such that $u \in V_i$. For each $v \in V(G)$, let $H_v$ be the subgraph of $G$ induced by the vertices at distance at most $d$ from $v$, and let $\mathcal{H} = \{H_v\}_{v\in V(G)}$. Let $\mathcal{T'} = (T,\{X'_t\}_{t\in V(T)})$, with $X'_t = \{v \in V(G) : V(H_v) \cap X_t \neq \varnothing\}$ for each $t \in V(T)$. By \cite[Lemma~6.1]{DaMS22}, $\mathcal{T}'$ is a tree decomposition of $G(\mathcal{H})$ and hence, by \Cref{sametreeA}, of $G^{1+2d}$ as well. Moreover, for each $v \in V(G)$, $V(H_v) \cap X_t \neq \varnothing$ if and only if $v$ is at distance at most $d$ from $X_t$ in $G$ and the set of all such vertices $v$ can be computed using BFS in $O(|V(G)| + |E(G)|)$ time. Therefore, $\mathcal{T}'$ can be computed in $O(|V(T)| \cdot (|V(G)| + |E(G)|)$ time. For each $1 \leq i \leq \lceil \frac{m}{1+2d} \rceil$, let now $V'_i = \bigcup_{(1+2d)(i-1) < j \leq (1+2d)i}V_j$. We claim that $(V'_1,\ldots,V'_{\lceil \frac{m}{1+2d} \rceil})$ is a layering of $G^{1+2d}$. Clearly, these sets partition $V(G^{1+2d})$. Moreover, for each edge $uv \in E(G^{1+2d})$, we have that $d_{G}(u,v) \leq 1+2d$ and so $|l(i)-l(j)| \leq 1+2d$. Consequently, $u$ and $v$ belong to either the same $V'_{i}$ or to consecutive $V'_{i}$'s. 

We now show that, for each $1 \leq i \leq \lceil \frac{m}{1+2d} \rceil$ and $t\in V(T)$, $\alpha(G^{1+2d}[V'_i \cap   X'_t]) \leq (1+4d)k$. Suppose, to the contrary, that $\alpha(G^{1+2d}[V'_i \cap X'_t]) > (1+4d)k$ for some $i$ and $t$ as above. Then, there exists an independent set $U = \{{u_1},\ldots,{u_{(1+4d)k+1}}\}$ of $G^{1+2d}$ contained in $V'_i \cap X'_t$. By construction, $V(H_{u_p}) \cap X_t \neq \varnothing$ for each $1 \leq p \leq (1+4d)k+1$ and, for each such $p$, we pick an arbitrary vertex in $V (H_{u_p}) \cap X_t$ and denote it by $r(u_p,t)$. Note that, for $p \neq q$, $r(u_p,t)$ is distinct from and non-adjacent to $r(u_q,t)$ in $G$, for otherwise either $H_{u_p}$ and $H_{u_q}$ share a vertex or there is an edge connecting them in $G$, from which $u_pu_q \in G(\mathcal{H}) = G^{1+2d}$, contradicting the fact that $U$ is an independent set of $G^{1+2d}$. Now, by construction, each $r(u_p,t)$ belongs to $V(H_{u_p})$ and so is at distance at most $d$ in $G$ from $u_p$. Moreover, $u_p$ belongs to the layer $V'_i$ of $G^{1+2d}$, for each $1\leq p \leq (1+4d)k + 1$. Therefore, $(1+2d)(i-1) < l(u_p) \leq (1+2d)i$ and $(1+2d)(i-1) - d < l(r(u_p,t)) \leq (1+2d)i + d$, for each $1\leq p \leq (1+4d)k + 1$. That is, each $r(u_p,t)$ belongs to one of the $1+4d$ consecutive layers $V_{(1+2d)(i-1) - d + 1},\ldots, V_{(1+2d)i + d}$ of $G$. Since $|U| > (1+4d)k$, at least one such layer $V_r$ contains $k+1$ vertices of the form $r(u_p,t)$. Therefore, $\alpha(G[V_r \cap X_t]) > k$, a contradiction.
\end{proof}

\begin{lemma} Fix an even $k \in \mathbb{N}$. There exist graphs $G$ with layered tree-independence number $1$ and such that the layered tree-independence number of $G^k$ is arbitrarily large. 
\end{lemma}

\begin{proof} By the proof of \cite[Proposition~3.7]{MR22}, for every graph $H$, there exists a chordal graph $G$ such that $G^k$ contains an induced subgraph isomorphic to $H$. Take $H = K_{5n,5n}$ and one such $G$. By \Cref{layeredKnnA}, the layered tree-independence number of $G^k$ is at least $n$, whereas $\tin(G) = 1$ and hence $G$ has layered tree-independence number $1$.
\end{proof}


\subsection{Intersection graphs of bounded layered tree-independence number}\label{sec:layeredlemmas}

In this section, we show that classes of intersection graphs of unit disks in $\mathbb{R}^2$ and of paths with bounded horizontal part on a grid both have bounded layered tree-independence number. As it will appear from the proofs, our tree decompositions witnessing this are in fact path decompositions. 

\begin{theorem}\label{unitdisklayeredA} Let $G$ be the intersection graph of a family $\mathcal{D}$ of $n$ unit disks. It is possible to compute, in $O(n)$ time, a tree decomposition $\mathcal{T} = (T,\{X_t\}_{t\in V(T)}\})$ and a layering $(V_1, V_2, \ldots)$ of $G$ such that $|V(T)| = O(n)$ and, for each bag $X_t$ and layer $V_i$, $\alpha(G[X_t \cap V_i]) \leq 8$. In particular, $G$ has layered tree-independence number at most $8$. 
\end{theorem}

\begin{proof} Without loss of generality, the common radius of the disks is $1/2$, the collection $\mathcal{D}$ is contained in the positive quadrant, and $G$ is connected. Therefore, $\mathcal{D}$ is contained in a $(m-1) \times (m-1)$ square with $m = O(n)$. For each $i \in \mathbb{N}$, let $C_i = \{(x,y) \in \mathbb{R}^2 : i-1 \leq x \leq i\}$ and $R_i = \{(x,y) \in \mathbb{R}^2 : i-1 \leq y \leq i\}$ be the \textit{$i$-th vertical strip} and \textit{$i$-th horizontal strip}, respectively. For each $v \in V(G)$, the disk $D_v$ intersects at most two horizontal strips. Fix an arbitrary $j$ such that the horizontal strip $R_j$ intersects $D_v$ in a region with area at least $\pi/8$ and let $r(v) = j$. We first construct a tree decomposition of $G$. Consider a path $T$ with $m+1$ vertices $\{t_1,\ldots,t_{m+1}\}$. Let $X_{t_i} = \{v \in V(G) : D_v \cap C_i \neq \varnothing\}$. Clearly, for each $v\in V(G)$, there exists $i$ with $v \in X_{t_i}$. Let now $uv \in E(G)$. Then, there exists a point $(x, y) \in \mathbb{R}^2$ contained in both $D_u$ and $D_v$ and so $\{u, v\} \subseteq X_{t_{\lceil x \rceil}}$. Finally, for each $v \in V(G)$, the vertical strips intersecting $D_v$ are consecutive, hence the set of nodes of $T$ whose bag contains $v$ is a subpath. This shows that $\mathcal{T} = (T,\{X_{t_i}\}_{1 \leq i \leq m})$ is a tree decomposition of $G$. It is easy to see that this tree decomposition can be computed in linear time. We now construct a layering of $G$ as follows. For each $j \in \mathbb{N}$, let $V_j =\{v \in V(G) : r(v) = j\}$. Clearly, $(V_1,\ldots,V_{m+1})$ is a partition of $V(G)$. Observe that, for $i$ and $j$ with $|i-j| \geq 2$, if $u\in V_i$ and $v\in V_j$, then $D_u \cap D_v = \varnothing$ and so $uv \not\in E(G)$. Clearly, the layering $(V_1,\ldots,V_{m+1})$ can be computed in linear time. 

Consider an arbitrary bag $X_{t_i}$ and layer $V_j$ as defined above. For any $v \in X_{t_i} \cap V_j$, we have that $r(v)=j$ and $D_v \cap C_i \neq \varnothing$. Then, the horizontal strip $R_j$ intersects $D_v$ in a region with area at least $\pi/8$ and $D_v$ intersects at most one of $C_{i-1}$ and $C_{i+1}$. Therefore, $D_v$ intersects $R_j \cap (C_{i-1} \cup C_i \cup C_{i+1})$ in a region with area at least $\pi/8$. Now, if two vertices in $X_{t_i} \cap V_j$ are non-adjacent in $G$, then the corresponding unit disks are disjoint. Since the area of $R_j \cap (C_{i-1} \cup C_i \cup C_{i+1})$ is $3$, there are at most $\frac{3}{\pi/8} \leq 8$ pairwise non-adjacent vertices in $G[X_{t_i} \cap V_j]$, as claimed.
\end{proof}

\begin{theorem}\label{layeredVPGA}
Let $G$ be a graph on $n$ vertices together with a grid representation $\mathcal{R} = (\mathcal{G}, \mathcal{P},x)$ such that each path in $\mathcal{P}$ has horizontal part of length at most $\ell-1$, for some fixed $\ell \geq 1$, and number of bends constant. It is possible to compute, in $O(n^2)$ time, a tree decomposition $\mathcal{T} = (T, \{X_t\}_{t\in V(T)})$ and a layering $(V_1,V_2\ldots)$ of $G$ such that $|V(T)| = O(n^2)$ and, for each bag $X_t$ and layer $V_i$, $\alpha(G[X_t \cap V_i]) \leq 4\ell-1$. In particular, $G$ has layered tree-independence number at most $4\ell-1$. 
\end{theorem}

\begin{proof} We assume without loss of generality that $G$ is connected. Denote by $P_v$ the path corresponding to the vertex $v \in V(G)$. For each $u \in V(G)$ and $v \in N_G(u)$, we find an arbitrary grid-point that belongs to both $P_u$ and $P_v$ in $O(1)$ time (since $uv \in E(G)$, such a grid-point exists). Let $p_{uv}$ be the projection of this grid-point onto the $y$-axis. Let $T$ be the vertical path whose nodes are the grid-points of the form $p_{uv}$. Clearly, $|V(T)| = O(n^2)$. Order the nodes of $T$ by decreasing $y$-coordinates of the corresponding grid-points. For each $t \in V(T)$, let $X_{t}' = \{u \in V(G) : p_{uv} = t \ \mbox{for some}\ v \in V(G)\}$. We finally update $\{X_t'\}_{t\in V(T)}$ as follows. For each $v \in V(G)$, let $t_{v}^{\min}$ and $t_{v}^{\max}$ be the smallest and largest node of $T$ such that the corresponding bag contains $v$, respectively. We then add $v$ to all the bags corresponding to nodes between $t_{v}^{\min}$ and $t_{v}^{\max}$. Let $\{X_t\}_{t\in V(T)}$ be the family of bags thus obtained and let $\mathcal{T} = (T, \{X_t\}_{t\in V(T)})$. It is easy to see that $\mathcal{T}$ can be constructed in $O(n^2)$ time. We show that $\mathcal{T}$ is a tree decomposition of $G$. Since $G$ has no isolated vertex, for each $u \in V(G)$, there exists $t \in V(T)$ such that $t = p_{uv}$ for some $v \in V(G)$. Hence, $u \in X_t$ and (T1) holds. If $uv \in E(G)$, then $P_u$ and $P_v$ share at least one grid-point of $\mathcal{G}$ and so there exists $t \in V(T)$ such that $t = p_{uv}$ and (T2) holds. Finally, by construction, the set $\{t \in V(T) : u \in X_t\}$ induces a subpath of $T$ and (T3) holds. 
Let $C_i = \{(x,y) \in \mathbb{R}^2 : 2(i-1)\ell \leq x \leq 2i\ell\}$ be the \textit{$i$-th vertical strip}. For each $v \in V(G)$ fix an arbitrary $j$ such that the vertical strip $C_j$ intersects $P_v$ and let $r(v) = j$. 
 We now construct a layering of $G$ as follows. For each $i \in \mathbb{N}$, let $V_i = \{v \in V(G): r(v) = i\}$. Clearly, $(V_1,V_2, \ldots)$ is a partition of $V(G)$. Observe that, for $i$ and $j$ with $|i-j| \geq 2$, if $u\in V_i$ and $v\in V_j$, then $P_u \cap P_v = \varnothing$ and so $uv \not\in E(G)$. The layering can be clearly computed in linear time given the representation of $G$. 

Consider an arbitrary bag $X_t$ and layer $V_i$ as defined above. We now show that $\alpha(G[X_t]\cap V_i) \leq 4\ell - 1$. Let $I$ be an independent set of $G[X_t \cap V_i]$. Observe that, by construction, for each $v \in I \subseteq X_t \cap V_i$, the path $P_v$ contains a grid-point at the intersection of the row of $\mathcal{G}$ indexed by $t$ and a column of $\mathcal{G}$ indexed by $p$, for some $ 2(i-1)\ell - (\ell - 1) \leq p \leq 2i\ell + (\ell - 1)$. Since $I$ is an independent set, each grid-point on the row indexed by $t$ belongs to at most one $P_v$ with $v \in I$; and since there are $4\ell -1$ columns of $\mathcal{G}$ indexed by some $2(i-1)\ell - (\ell-1) \leq p \leq 2i\ell + (\ell-1)$, we conclude that $|I| \leq 4\ell - 1$.
\end{proof}


\section{Fractional $\tin$-fragility}\label{sec:fragilityA}

Let $p$ be a width parameter in $\{\tw, \tin\}$. Fractional $\tw$-fragility was first defined in \cite{Dvo16}. We provide here an equivalent definition from \cite{Dvo22}, which was explicitly extended to the case $p = \tin$ in \cite{GWP22}.

\begin{definition}
For $\beta \leq 1$, a $\beta$-general cover of a graph $G$ is a multiset $\mathcal{C}$ of subsets of $V(G)$ such that each vertex belongs to at least $\beta|\mathcal{C}|$ elements of the cover. The $p$-width of the cover is $\max_{C \in \mathcal{C}}p(G[C])$.

For a parameter $p$, a graph class $\mathcal{G}$ is fractionally $p$-fragile if there exists a function $f\colon\mathbb{N}\rightarrow\mathbb{N}$ such that, for every $r\in\mathbb{N}$, every $G \in \mathcal{G}$ has a $(1 - 1/r)$-general cover with $p$-width at most $f(r)$. 

A fractionally $p$-fragile class $\mathcal{G}$ is efficiently fractionally $p$-fragile if there exists an algorithm that, for every $r\in\mathbb{N}$ and $G \in \mathcal{G}$, returns in $\mathsf{poly}(|V(G)|)$ time a $(1 - 1/r)$-general cover $\mathcal{C}$ of $G$ and, for each $C \in \mathcal{C}$, a tree decomposition of $G[C]$ of width (if $p = \tw$) or independence number (if $p = \tin$) at most $f(r)$, for some function $f\colon\mathbb{N}\rightarrow\mathbb{N}$.
\end{definition}

Note that classes of bounded tree-independence number are efficiently fractionally $\tin$-fragile thanks to \cite{DFGK22}. Hence, the family of efficiently fractionally $\tin$-fragile classes contains the two incomparable families of bounded tree-independence number classes and efficiently fractionally $\tw$-fragile classes (to see that they are incomparable, consider chordal graphs and planar graphs). We now identify one more subfamily:

\begin{lemma}\label{layeredtofragileA} Let $\ell \in \mathbb{N}$ and let $G$ be a graph. For each $r \in \mathbb{N}$, given a tree decomposition $\mathcal{T} = (T,\{X_t\}_{t\in V(T)})$ of $G$ and a layering $(V_0,V_1,\ldots)$ of $G$ such that, for each bag $X_t$ and layer $V_i$, $\alpha(G[X_t \cap V_i]) \leq \ell$, it is possible to compute in $O(|V(G)|)$ time a $(1 - 1/r)$-general cover $\mathcal{C}$ of $G$ and, for each $C \in \mathcal{C}$, a tree decomposition of $G[C]$ with tree-independence number at most $\ell (r-1)$. In particular, if every graph in a class $\mathcal{G}$ has layered tree-independence number at most $\ell$, then $\mathcal{G}$ is fractionally $\tin$-fragile with $f(r) = \ell (r-1)$.
\end{lemma}

\begin{proof} Fix $r \in \mathbb{N}$. Let $\mathcal{T} = (T, \{X_t\}_{t\in V(T)})$ and $(V_0,V_1,\ldots)$ be the given tree decomposition and layering of $G$, respectively. For each $m \in \{0,\ldots,r-1\}$, let $C_m = \bigcup_{i \not\equiv m \pmod r}V_i$. We claim that $\mathcal{C} = \{C_m : 0 \leq m \leq r-1\}$ is a $(1-1/r)$-general cover of $G$ with tree-independence number at most $\ell (r-1)$. Observe first that each $v \in V(G)$ is not covered by exactly one element of $\mathcal{C}$ and so it belongs to $r-1 = (1-1/r)|\mathcal{C}|$ elements of $\mathcal{C}$. Let now $C \in \mathcal{C}$. Each component $K$ of $G[C]$ is contained in at most $r-1$ (consecutive) layers and so, since $\alpha(G[X_t \cap V_i]) \leq \ell$ for each bag $X_t$ and layer $V_i$, restricting the bags in $\mathcal{T}$ to $V(K)$, gives a tree decomposition of $K$ with tree-independence number at most $\ell (r-1)$. We then merge the tree decompositions of the components of $G[C]$ into a tree decomposition of $G[C]$ with tree-independence number at most $\ell (r-1)$ in linear time.    
\end{proof}

Note that the same argument of \Cref{layeredtofragileA} shows that, if every graph in a class $\mathcal{G}$ has bounded layered treewidth, then $\mathcal{G}$ is fractionally $\tw$-fragile. The following result implies that, if a class is fractionally $\tin$-fragile, then it has bounded biclique number.

\begin{theorem}\label{bicliqueA}
For any function $f\colon \mathbb{N} \rightarrow \mathbb{N}$ and integer $r >2$, there exists $n$ such that no $(1-1/r)$-general cover of $K_{n,n}$ has tree-independence number less than $f(r)$. Hence, the class $\{K_{n,n}: n \in \mathbb{N}\}$ is not fractionally $\tin$-fragile.
\end{theorem}

\begin{proof}
Fix arbitrary $f\colon \mathbb{N} \rightarrow \mathbb{N}$ and $r>2$. Consider a copy $G$ of $K_{n,n}$, with $n > f(r)/(1 - 2/r)$. Let $\mathcal{C}$ be a $(1-1/r)$-general cover of $G$. Then, every vertex of $G$ belongs to at least $(1-1/r)|\mathcal{C}|$ elements of $\mathcal{C}$ and so there exists $C \in \mathcal{C}$ of size at least $2n(1-1/r)$. Let $A$ and $B$ be the two bipartition classes of $G$. Then, $|A\cap C| \geq |C| - |B| \geq 2n(1-1/r) - n = n(1-2/r) > f(r)$ and, similarly, $|B\cap C| > f(r)$. Therefore, $G[C]$ contains $K_{f(r),f(r)}$ as an induced subgraph and since $\tin(K_{f(r),f(r)}) = f(r)$ \cite{DaMS22}, $\tin(G[C]) \geq f(r)$. 
\end{proof}

However, the following result shows that small biclique number does not guarantee fractional $\tin$-fragility.

\begin{theorem}
The class of $K_{2,3}$-free graphs is not fractionally $\tin$-fragile.
\end{theorem}

\begin{proof} Let $G_n$ be the $n$-dimensional grid graph of width $n$, i.e., the graph with vertex set $V(G_n) = [n]^n = \{(a_1,\ldots,a_n) : 1 \leq a_1,\ldots,a_n \leq n\} $, where two vertices $(a_1,\ldots,a_n)$ and $(b_1,\ldots,b_n)$ are adjacent if and only if $\sum_{1 \leq i \leq n} |a_i-b_i|=1$. It is not difficult to see that $G_n$ is $K_{2,3}$-free, for each $n \in \mathbb{N}$. We show that the class $\{G_n : n\in\mathbb{N}\}$ is not fractionally $\tin$-fragile. 

Fix arbitrary $f\colon \mathbb{N}\rightarrow\mathbb{N}$ and $r>2$. For such a choice, fix $n \in \mathbb{N}$ such that $\frac{r-4}{2r}n + 1 \geq R(3, f(r))$, where $R(3,s)$ denotes the smallest integer $m$ for which every graph on $m$ vertices either contains a clique of size $3$ or an independent set of size $s$. We now show that every $(1-1/r)$-general cover of $G_n$ has tree-independence at least $f(r)$. Let $\mathcal{C}$ be a $(1-1/r)$-general cover of $G_n$. Then, every vertex of $G_n$ belongs to at least $(1-1/r)|\mathcal{C}|$ elements of $\mathcal{C}$ and so there exists $C \in \mathcal{C}$ containing at least $(1-1/r)|V(G_n)| = (1-1/r)n^n$ vertices of $G_n$. Fix such a $C$ and let $G$ be the subgraph of $G_n$ induced by $C$. We claim that $\tin(G) \geq f(r)$.

Observe first that, for each $v \in V(G_n)$, $n \leq d_{G_n}(v) \leq 2n$. Hence, $2|E(G_n)| = \sum_{v \in V(G_n)}d_{G_n}(v) \geq n\cdot n^n$. Consider now the graph $G'$ obtained from $G_n$ by deleting the vertex set $C$. Clearly, $G'$ has at most $n^n/r$ vertices. Since deleting a vertex from $G_n$ decreases the number of edges of the resulting graph by at most $2n$, we have that $|E(G)|\geq |E(G_n)| - 2n|V(G')|$, from which $\sum_{v \in V(G)}d_{G}(v) \geq n\cdot n^n - 2\cdot 2n\cdot n^n/r = n\cdot n^n(1 - 4/r)$. Therefore, the average degree of $G$ is at least $n(1-4/r)$ and so $\tw(G) \geq \frac{r-4}{2r}n$, for example by \cite[Corollary~1]{CS05}. This implies that every tree decomposition of $G$ has a bag of size at least $\frac{r-4}{2r}n + 1 \geq R(3, f(r))$ and, since $G$ is triangle-free, it follows that $\tin(G) \geq f(r)$.
\end{proof}


\subsection{Intersection graphs of fat objects}\label{fatA}

In this section we show that the class of intersection graphs of fat objects in $\mathbb{R}^d$ is efficiently fractionally $\tin$-fragile. Let $d \geq 2$ be a fixed integer. A \textit{box of size $r$} is an axis-aligned hypercube in $\mathbb{R}^d$ of side length $r$. The \textit{size} of an object $O$ in $\mathbb{R}^d$, denoted $s(O)$, is the side length of its smallest enclosing hypercube. Unless otherwise stated, all boxes considered in the following are axis-aligned. 

Chan \cite{Cha03} considered the following definition of fatness: A collection of objects in $\mathbb{R}^d$ is $\textit{fat}$ if, for any $r$ and size-$r$ box $R$, we can choose $c$ points in $\mathbb{R}^d$ such that every object that intersects $R$ and has size at least $r$ contains at least one of the chosen points. Chan also stated that a collection of balls or boxes with bounded aspect ratios are fat (recall that the aspect ratio of a box is the ratio of its largest side length over its smallest side length). We slightly generalize this fatness definition as follows.

\begin{definition} A collection of objects in $\mathbb{R}^d$ is $c$-$\textit{fat}$ if, for any $r$ and any size-$r$ closed box $R$, for every sub-collection $\mathcal{P}$ of pairwise non-intersecting objects which intersect $R$ and are of size at least $r$, we can choose $c$ points in $\mathbb{R}^d$ such that every object in $\mathcal{P}$ contains at least one of the chosen points.
\end{definition}

\begin{remark} When working with a $c$-fat collection of objects, we assume that some reasonable operations can be done in constant time: determining the center and size of an object, deciding if two objects intersect and constructing the geometric realization of the collection. 
\end{remark}

For completeness, we now observe that collections of balls and boxes with bounded aspect ratios are $c$-fat. The following lemma comes in handy. For a measurable set $A \subseteq \mathbb{R}^d$, we denote by $\vol(A)$ the Lebesgue measure of $A$.

\begin{lemma}\label{cfatA} Let $\mathcal{O}$ be a finite collection of measurable objects in $\mathbb{R}^d$. Let $r_0 = \min_{O \in \mathcal{O}}s(O)$. Suppose that there exists $0 < a \leq 1$ such that, for every object $O \in \mathcal{O}$, for every point $x \in O$ and for every $r$ with $r_0 \leq r \leq s(O)$, the set $O' \subseteq O$ of points within distance $r$ from $x$ has $\vol(O') \geq ar^d$. Then, $\mathcal{O}$ is $\frac{3^d}{a}$-fat.
\end{lemma}

\begin{proof} Fix $r$ and a size-$r$ closed box $R$. Consider the box $R'$ of side length $3r$ and with the same center as $R$. Any object in $\mathcal{O}$ of size at least $r$ intersecting $R$ must intersect $R'$ in a set with volume at least $ar^d$. However, $\vol(R') = (3r)^d$, and so there are at most $\frac{(3r)^d}{ar^d} = \frac{3^d}{a}$ pairwise non-intersecting objects from $\mathcal{O}$ which intersect $R$. 
\end{proof}

Since any size-$r$ ball has volume at least $r^d/d!$, a collection of balls in $\mathbb{R}^d$ is $3^d d!$-fat. Moreover, since the length of a main diagonal of a $d$-dimensional box of side length $l$ is $l\sqrt{d}$, a size-$r$ box with aspect ratio at most $t$ has volume at least $(\frac{r}{t\sqrt{d}})^d$, and so a collection of boxes (not necessarily axis-aligned) in $\mathbb{R}^d$ with aspect ratio at most $t$ is $(3t\sqrt{d})^d$-fat. 

\begin{theorem}\label{treealphafatA} Let $\mathcal{O}$ be a $c$-fat collection of objects in $\mathbb{R}^d$ and let $G$ be its intersection graph. For each $r_0 > 1$, let $f(r_0) = 2\Big\lceil \frac{1}{1-\big(1-\frac{1}{r_0}\big)^{\frac{1}{d}}}\Big\rceil$. Then, we can compute in linear time a $(1-1/r_0)$-general cover $\mathcal{C}$ of $G$ of size at most $(f(r_0)/2-1)^d$. Moreover, for each $C \in \mathcal{C}$, we can compute in linear time a tree decomposition $\mathcal{T} = (T,\{X_t\}_{t\in V(T)})$ of $G[C]$, with $|V(T)| \leq |V(G)|+1$, such that $\alpha(\mathcal{T}) \leq cf(r_0)^{2d}$. 
\end{theorem}

Before proving \Cref{treealphafatA}, we outline the idea in the case of disk graphs in $\mathbb{R}^2$. Suppose first that we are trying to find a $(1-1/r_0)$-general cover $\mathcal{C}$ of bounded tree-independence number of a unit disk graph. We can build each element of the cover starting from an appropriate grid in the plane as follows. Suppose that $\mathcal{H}(y)$ is a grid in $\mathbb{R}^2$, indexed by some $y \in \mathbb{R}^2$, splitting the plane into a collection $\mathcal{B}$ of squares of side length $2r_0$. We first discard all disks intersecting $\mathcal{H}(y)$. The vertices corresponding to the remaining disks will form the element $C(y)$ of the cover. We can obtain a tree decomposition for the subgraph induced by $C(y)$ as follows. We add a node for each square $B \in \mathcal{B}$ and associate to this node a bag containing precisely the vertices whose corresponding disks lie in $B$. We then connect the nodes appropriately to obtain a tree. The resulting tree decomposition will have small tree-independence number since, inside each square $B$, there are at most $4r_{0}^2$ pairwise non-intersecting disks. Shifting the grid $\mathcal{H}(y)$ around the plane via the vector $y$ and proceeding as above will ensure that every vertex of the unit disk graph is covered by sufficiently many elements. 

The situation is more challenging if disks have different radiuses. When both large and small disks occur, if the grids are too dense (i.e., they divide the plane into very small squares), then large disks will not belong to most elements of the cover, whereas if the grids are too sparse (i.e., they divide the plane into very large squares), then there might be too many pairwise non-intersecting small disks inside each square. To resolve this, we use an idea from \cite[Theorem~4]{DGLTT22}. Specifically, we sort disks into different ranks according to their radius, so that the larger the radius the smaller the rank. Large disks will be ``covered'' by sparse grids, whereas small disks will be ``covered'' by dense grids. For each possible value $i$ of the rank, we will consider grids of rank $i$ arising in a quadtree-like manner from a fixed rank-$0$ grid (a sparesest grid), and we will discard rank-$i$ disks intersecting rank-$i$ grids. The vertices corresponding to the remaining disks will form an element $C(y)$ of the cover. We will then add a node for each square $B_i$ induced by the rank-$i$ grid, and associate to this node a bag containing precisely the vertices whose corresponding disks intersect $B_i$ and have rank at most $i$. Finally, for each node $t_i$ corresponding to a rank-$i$ square $B_i$, we will add the edge $t_it_j$ if $t_j$ correspond to the rank-$j$ square $B_j$, with $j > i$, such that $B_j$ is contained in $B_i$. As before, we will shift the grids around the plane via $y$ to ensure that we obtain indeed a general cover. 

\begin{proof}[Proof of \Cref{treealphafatA}] In this proof, $[n]$ denotes the set $\{0, 1, \ldots, n\}$. Fix an arbitrary $r_0 > 1$. In the following, for ease of notation, we simply let $r:= f(r_0) = 2\Big\lceil \frac{1}{1-\big(1-\frac{1}{r_0}\big)^{\frac{1}{d}}}\Big\rceil$, and denote by $O_v$ the object corresponding to the vertex $v \in V(G)$. By possibly rescaling, we may assume that each object in the collection $\mathcal{O}$ has size at most $1$. For each $v \in V(G)$, define the \textit{rank} of the object $O_v$ as the quantity $\rk(O_v) = \lfloor \log_{\frac{1}{r}}s(O_v)\rfloor$. Let $k_0 = \max_{v \in V(G)}\rk(O_v)$. For each $0 \leq i \leq k_0$, $1 \leq j \leq d$ and $y = (y_1,\ldots,y_d) \in [\frac{r}{2}-1]^d$, let $\mathcal{H}^i_j(y)$ be the set of points in $\mathbb{R}^d$ whose $j$-th coordinate is equal to $m_j(\frac{1}{r})^{i-1}+y_j\sum_{k={i}}^{k_0+1}(\frac{1}{r})^{k}$, for some $m_j \in \mathbb{Z}$, and let $\mathcal{H}^i(y) = \bigcup_{1\leq j \leq d}\mathcal{H}^i_j(y)$. Moreover, let $V^i = \{v \in V(G) : \rk(O_v)=i\}$, $C^i(y) = \{v \in V^i: O_v \cap \mathcal{H}^i(y) = \varnothing\}$, and $C(y) = \bigcup_{0 \leq i \leq k_0}C^i(y)$. 

\begin{claim}
$\mathcal{C} = \{C(y) : y \in [\frac{r}{2}-1]^d\}$ is a $(1-1/r_0)$-general cover of $G$ of size $(f(r_0)/2-1)^d$.
\end{claim}

\begin{claimproof} For each $1\leq j \leq d$, let $e_j$ be the unit vector in $\mathbb{R}^d$ whose $j$-th coordinate is $1$. Let $v\in V(G)$ and let $i = \rk(O_v)$. Then, $v \in C(y)$ for some $y \in [\frac{r}{2}-1]^d$, if and only if $O_v$ does not intersect $\mathcal{H}^i(y)$ and the latter happens if and only if $O_v$ does not intersect $\mathcal{H}^i_j(y)$ for any $1 \leq j \leq d$. Note that $\mathcal{H}^i_j(y)$ is a collection of hyperplanes in $\mathbb{R}^d$, which are orthogonal to the $j$-th axis and at pairwise distance at least $(\frac{1}{r})^{i-1} = r(\frac{1}{r})^{i}$. Moreover, for any $1\leq j\leq d-1$, $\mathcal{H}^i_{j}(y + e_j)$ can be obtained by shifting $\mathcal{H}^i_j(y)$ along the $j$-th axis of a quantity $\sum_{k={i}}^{k_0+1}(\frac{1}{r})^{k}$, and it is easy to see that, since $r \geq 2$, we have $(\frac{1}{r})^{i} < \sum_{k={i}}^{k_0+1}(\frac{1}{r})^{k} < 2(\frac{1}{r})^{i}$. On the other hand, since $\rk(O_v)=i$, it follows from the definition of rank that $O_v$ can be enclosed in a box of size $(\frac{1}{r})^{i}$. We now count the number of points $y \in [\frac{r}{2}-1]^d$ such that $O_v$ does not intersect $\mathcal{H}^i(y)$. For fixed $1 \leq j \leq d$, the previous observations imply that there is at most one value of $y_j$ for which a point $y \in [\frac{r}{2}-1]^d$ is such that $\mathcal{H}^i_j(y) \cap O_v \neq \varnothing$, so at least $\frac{r}{2} - 1$ values of $y_j$ for which $y$ is such that $\mathcal{H}^i_j(y) \cap O_v = \varnothing$. Therefore, there are at least $(\frac{r}{2}-1)^d$ points $y \in [\frac{r}{2}-1]^d$ such that $\mathcal{H}^i(y)$ does not intersect $O_v$. Since the set $[\frac{r}{2}-1]^d$ has size $(\frac{r}{2})^d$, the proportion of elements of $\mathcal{C}$ containing $v$ is at least $\frac{(\frac{r}{2}-1)^d}{(\frac{r}{2})^d} = (1-\frac{2}{r})^d \geq ((1-\frac{1}{r_0})^{\frac{1}{d}})^d = 1-\frac{1}{r_0}$, as claimed.
\end{claimproof}

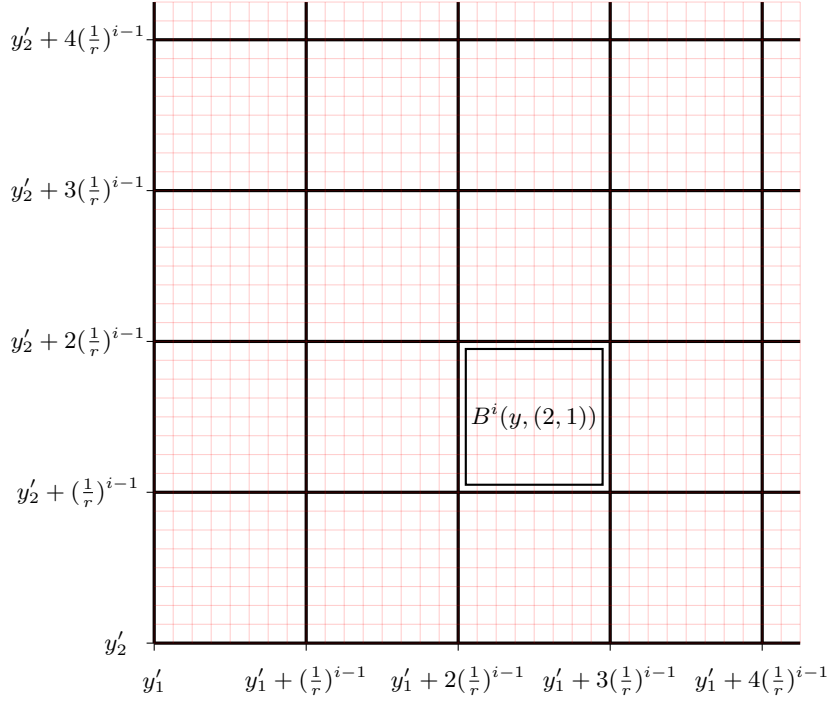
\begin{figure}
\centering
\begin{tikzpicture}
\draw (0,0) -- (0,-.1);
\draw (2,0) -- (2,-.1);
\draw (4,0) -- (4,-.1);
\draw (6,0) -- (6,-.1);
\draw (8,0) -- (8,-.1);

\draw (0,0) -- (-.1,0);
\draw (0,2) -- (-.1,2);
\draw (0,4) -- (-.1,4);
\draw (0,6) -- (-.1,6);
\draw (0,8) -- (-.1,8);

\draw[very thick] (0,0) -- (0,8.5);
\draw[very thick] (2,0) -- (2,8.5);
\draw[very thick] (4,0) -- (4,8.5);
\draw[very thick] (6,0) -- (6,8.5);
\draw[very thick] (8,0) -- (8,8.5);

\draw[very thick] (0,0) -- (8.5,0);
\draw[very thick] (0,2) -- (8.5,2);
\draw[very thick] (0,4) -- (8.5,4);
\draw[very thick] (0,6) -- (8.5,6);
\draw[very thick] (0,8) -- (8.5,8);

\node[draw= none] at (0,-.5) {\small $y'_1$};
\node[draw= none] at (2,-.5) {\small $y'_1 + (\frac{1}{r})^{i-1}$};
\node[draw= none] at (4,-.5) {\small $y'_1 + 2(\frac{1}{r})^{i-1}$};
\node[draw= none] at (6,-.5) {\small $y'_1 + 3(\frac{1}{r})^{i-1}$};
\node[draw= none] at (8,-.5) {\small $y'_1 + 4(\frac{1}{r})^{i-1}$};

\node[draw= none] at (-.5,0) {\small $y'_2$};
\node[draw= none] at (-1,2) {\small $y'_2 + (\frac{1}{r})^{i-1}$};
\node[draw= none] at (-1,4) {\small $y'_2 + 2(\frac{1}{r})^{i-1}$};
\node[draw= none] at (-1,6) {\small $y'_2 + 3(\frac{1}{r})^{i-1}$};
\node[draw= none] at (-1,8) {\small $y'_2 + 4(\frac{1}{r})^{i-1}$};

\foreach \i in {0,...,34}{
\pgfmathsetmacro{\x}{0.25*\i}
\draw[red,opacity=.2] (\x,0) -- (\x,8.5);
\draw[red,opacity=.2] (0,\x) -- (8.5,\x);
}

\draw[thick] (4.1,2.1) rectangle (5.9,3.9);
\node[draw=none] at (5,3) {\small $B^i(y,(2,1))$};
\end{tikzpicture}
\caption{$\mathcal{H}^i(y)$ (black) and $\mathcal{H}^{i+1}(y)$ (red) in dimension 2, where $y'_j = y_j \sum_{k=i}^{k_0+1} (\frac{1}{r})^k$.}
\end{figure}

Note that, for fixed $i$ and $y$ as above, the collection of hyperplanes $\mathcal{H}^i(y)$ splits the space into boxes of size $(\frac{1}{r})^{i-1}$. We now consider these boxes. For a vector $m = (m_1, \ldots, m_d) \in \mathbb{Z}^d$, consider the box $B^i(y,m) = \{(x_1,\ldots,x_d) \in \mathbb{R}^d : m_j(\frac{1}{r})^{i-1}+y_j\sum_{k={i}}^{k_0+1}(\frac{1}{r})^{k} < x_j < (m_j+1)(\frac{1}{r})^{i-1}+y_j\sum_{k={i}}^{k_0+1} (\frac{1}{r})^{k}, \ \mbox{for every} \ 1 \leq j \leq d\}$.

\begin{claim}\label{subboxA} For fixed $y \in [\frac{r}{2}-1]^d$ and any $0 \leq i' < i \leq k_0$, $\mathcal{H}^{i'}(y) \subseteq \mathcal{H}^{i}(y)$. Moreover, for any $m \in \mathbb{Z}^d$, the box $B^{i}(y,m)$ is completely contained in a box of the form $B^{i'}(y,m')$ for exactly one vector $m' \in \mathbb{Z}^{d}$. 
\end{claim}

\begin{claimproof} Let $y = (y_1, \ldots, y_d)$. Let $x=(x_1,\ldots,x_d) \in \mathcal{H}^{i'}(y)$. There exists $1 \leq j \leq d$ such that $x \in \mathcal{H}^{i'}_j(y)$. Then, there exists $m_j \in \mathbb{Z}$ such that 
\begin{align*}
x_j &=  m_j\Big(\frac{1}{r}\Big)^{i'-1}+y_j\sum_{k={i'}}^{k_0+1}\Big(\frac{1}{r}\Big)^{k} = m_j\Big(\frac{1}{r}\Big)^{i'-i}\Big(\frac{1}{r}\Big)^{i-1} +y_j\sum_{k={i'}}^{i-1}\Big(\frac{1}{r}\Big)^{k} +y_j\sum_{k={i}}^{k_0+1}\Big(\frac{1}{r}\Big)^{k} \\
&= \Big(m_j\Big(\frac{1}{r}\Big)^{i'-i}+y_j\sum_{k={i'}}^{i-1}\Big(\frac{1}{r}\Big)^{k-(i-1)}\Big)\Big(\frac{1}{r}\Big)^{i-1} + y_j\sum_{k={i}}^{k_0+1}\Big(\frac{1}{r}\Big)^{k}.
\end{align*}
Since the coefficient of $(\frac{1}{r})^{i-1}$ is an integer, we conclude that $x \in \mathcal{H}^{i}_j(y)$ and so $\mathcal{H}^{i'}(y) \subseteq \mathcal{H}^{i}(y)$.

To prove the remaining statement simply recall that, for $m, m' \in \mathbb{Z}^d$, $B^{i}(y,m)$ is one of the boxes induced by $\mathbb{R}^d \setminus \mathcal{H}^{i}(y)$ and $B^{i'}(y,m')$ is one of the boxes induced by $\mathbb{R}^d \setminus \mathcal{H}^{i'}(y)$. Since $\mathcal{H}^{i}(y)$ is a refinement of $\mathcal{H}^{i'}(y)$, for any box of the form $B^{i}(y,m)$ there must be exactly one box of the form $B^{i'}(y,m)$ containing it.
\end{claimproof}

We now construct a tree decomposition of $G[C(y)]$, for each element $C(y)$ of the $(1-1/r_0)$-general cover $\mathcal{C}$ defined above. Therefore, fix $y \in [\frac{r}{2}-1]^d$. For each $0 \leq i \leq k_0$ and $m \in \mathbb{Z}^d$, let $A^i(y,m) = \{v \in C^i(y): O_v \cap B^i(y,m) \neq \varnothing\}$ and let $X_{t^i(y,m)} = \bigcup_{0\leq k \leq i}\{v \in C^k(y) : O_v \cap B^i(y,m) \neq \varnothing\}$. In words, $X_{t^i(y,m)}$ is the set of vertices corresponding to objects of rank at most $i$ in $C(y)$ and intersecting the box $B^i(y,m)$. For each pair $(i,m)$, build a node $t^i(y,m)$ if $A^i(y,m) \neq \varnothing$ and associate to it the set $ X_{t^i(y,m)}$, which will be the corresponding bag in the tree decomposition we are building. We say that $t^{i_1}(y,{m_1})$ is a \textit{parent} of $t^{i_2}(y,m_2)$ if the following conditions are satisfied: $i_1 < i_2$, $B^{i_1}(y,m_1) \supseteq B^{i_2}(y,m_2)$ and, among all pairs satisfying these two conditions, $(i_1,m_1)$ has largest value of the first entry. Observe that, by \Cref{subboxA}, each node $t^{i_2}(y,{m_2})$ has at most one parent. For each pair of nodes $t^{i_1}(y,{m_1})$, $t^{i_2}(y,m_2)$ such that $t^{i_1}(y,{m_1})$ is a parent of $t^{i_2}(y,m_2)$, we then add the edge $t^{i_1}(y,{m_1})t^{i_2}(y,m_2)$. We claim that the resulting graph $F(y)$ is acyclic. Suppose, to the contrary, that it contains a cycle with vertices $t^{i_0}(y,{m_0}), t^{i_1}(y,m_1), \ldots, t^{i_{\ell-1}}(y,m_{\ell-1})$ in cyclic order. Without loss of generality, $t^{i_0}(y,{m_0})$ is a parent of $t^{i_1}(y,m_1)$. This implies that, for each $k$, $t^{i_k}(y,{m_k})$ is a parent of $t^{i_{k+1}}(y,m_{k+1})$ (indices modulo $\ell$). Therefore, by definition of parent, $i_0 < i_1 < \cdots < i_{\ell-1}$ and $i_{\ell-1} < i_0$, a contradiction.     

We then glue the components of $F(y)$ into a tree by adding a node $t^{-1}$ and making $t^{-1}$ adjacent to an arbitrary node of each component of $F(y)$. Let the resulting tree be $T(y)$. Observe that $|V(T(y))| \leq |V(G)|+1$. Indeed, $t^i(y,m)$ is a node of $T(y)$ only if $A^i(y,m) \neq \varnothing$ and, for fixed $y$, $A^{i_1}(y,m_1) \subseteq V(G)$ is disjoint from $A^{i_2}(y,m_2) \subseteq V(G)$. Let $\mathcal{T} = (T(y),\{X_{t^i(y,m)}\}_{t^i(y,m)\in V(T(y))})$, where we assign the empty bag to the node $t^{-1}$. Clearly, $\mathcal{T}$ can be computed in linear time.

\begin{claim} $\mathcal{T} = (T(y),\{X_{t^i(y,m)}\}_{t^i(y,m)\in V(T(y))})$ is a tree decomposition of $G[C(y)]$. 
\end{claim}

\begin{claimproof} We first check that (T1) holds. Let $v \in C(y)$ and suppose that $\rk(O_v) = i$. Then $O_v$ intersects one of the boxes $B^{i}(y,m)$, for some $m \in \mathbb{Z}^d$, and so $v \in A^i(y,m)$. Therefore, $t^i(y,m) \in V(T(y))$ and $v \in X_{t^i(y,m)}$. 

We now check that (T2) holds. Let $u,v \in C(y)$ such that $uv \in E(G)$. Then, $O_u \cap O_v \neq \varnothing$ and let $x=(x_1,\ldots,x_d)$ be a point in this intersection. Without loss of generality, $\rk(O_u) \leq \rk(O_v)=i$. Since $v \in C^i(y)$, $x \not \in \mathcal{H}^i(y)$. This implies that $x$ is contained in a box $B^i(y,m)$, for some $m\in\mathbb{Z}^d$, and so $v \in A^i(y,m)$, $t^i(y,m) \in V(T(y))$ and $\{u,v\} \subseteq X_{t^i(y,m)}$.  

We finally check that (T3) holds. For $v \in C(y)$, let $T(y)_v$ be the subgraph of $T(y)$ induced by the set of nodes of $T(y)$ whose bag contains $v$. Let $v \in C(y)$ and suppose that $\rk(O_v) = i$. Observe first that there is a unique $m \in \mathbb{Z}^d$ such that $v \in X_{t^{i}(y,m)}$, or else $O_v \cap \mathcal{H}^i(y) \neq \varnothing$ and $v \not \in C(y)$. Observe now that, by definition, $v \not \in X_{{t^{i_1}}(y,m_1)}$ for any $i_1 < i$ and $m_1 \in \mathbb{Z}^d$. Suppose finally that $v \in X_{t^{i_1}(y,m_1)}$, for some $i_1 > i$ and $m_1\in \mathbb{Z}^d$. Then, $O_v$ intersects $B^{i_1}(y,m_1)$ and, by \Cref{subboxA}, there is a unique $m' \in \mathbb{Z}^d$ such that $B^{i_1}(y,m_1)$ is completely contained in $B^{i}(y,m')$ (it is easy to see that $m' = m$). Hence, ${t^{i_1}(y,m_1)}$ must have a parent, say ${t^{i_2}(y,m_2)}$ for some $i_2$ such that $i_1 > i_2 \geq i$ and $m_2 \in \mathbb{Z}^d$. This means that $B^{i_2}(y,m_2) \supseteq B^{i_1}(y,m_1)$, and so $v \in X_{t^{i_2}(y,m_2)}$. We then deduce inductively that there must be a path from $t^{i_1}(y,m_1)$ to $t^{i}(y,m)$ in $T(y)_v$. Therefore, $T(y)_v$ is connected.
\end{claimproof} 

\begin{claim} $\alpha(\mathcal{T}) \leq cr^{2d}$.
\end{claim}

\begin{claimproof} Fix an arbitrary node $t^i(y,m)$ of $T(y)$. We bound the independence number of the subgraph of $G[C(y)]$ induced by the bag $X_{t^i(y,m)}$. Observe first that, for any $v \in X_{t^i(y,m)}$, $O_v$ intersects $B^i(y,m)$, which is a box of side length $(\frac{1}{r})^{i-1}$. Consider a collection $\mathcal{B}$ of $r^{2d}$ generic closed boxes in $\mathbb{R}^d$ of side length $(\frac{1}{r})^{i+1}$ and such that their union is exactly $B^i(y,m)$. Let $P \subseteq X_{t^i(y,m)}$ be an independent set of $G[C(y)]$ and let $\mathcal{P} = \{O_v : v \in P\}$ be the corresponding sub-collection of $\mathcal{O}$ of pairwise non-intersecting objects. For each $v \in P$, $\rk(O_v) \leq i$ and so $s(O_v) \geq (\frac{1}{r})^{i+1}$. Moreover, $O_v$ intersects at least one box from $\mathcal{B}$. Therefore, since the collection $\mathcal{O}$ is $c$-fat, we can choose $c$ points for each box in $\mathcal{B}$ in such a way that every object in $\mathcal{P}$ intersecting this box contains at least one of the chosen points. Let $C$ be the set of the $cr^{2d}$ chosen points. Since no two objects from $\mathcal{P}$ intersect, no two of them can contain the same point from $C$. Therefore, $|P| \leq cr^{2d}$, as claimed.
\end{claimproof}

This concludes the proof.
\end{proof}

\begin{corollary} There exist fractionally $\tin$-fragile classes of unbounded local tree-independence number. 
\end{corollary}

\begin{proof} Consider the class $\mathcal{G}$ of intersection graphs of disks in $\mathbb{R}^2$. By \Cref{treealphafatA}, $\mathcal{G}$ is fractionally $\tin$-fragile. Let now $G_n$ be the graph obtained from the $n\times n$-grid graph by adding a dominating vertex. By the proof of \Cref{equivlayeredA}, the class $\mathcal{G}' = \{G_n : n \in \mathbb{N}\}$ has unbounded local tree-independence number. Since $\mathcal{G'} \subseteq \mathcal{G}$, the result follows.
\end{proof}


\section{PTASes}\label{sec:ptasesA}

Let us begin by defining \textsc{Max Weight Independent Packing}. Given a graph $G$ and a finite family $\mathcal{H} = \{H_j\}_{j\in J}$ of connected non-null subgraphs of $G$, an \textit{independent $\mathcal{H}$-packing} in $G$ is a subfamily $\mathcal{H}' = \{H_i\}_{i\in I}$ of subgraphs from $\mathcal{H}$ (that is, $I \subseteq J$) that are at pairwise distance at least $1$, that is, they are vertex-disjoint and there is no edge between any two of them. If the subgraphs in $\mathcal{H}$ are equipped with a weight function $w\colon J \rightarrow \mathbb{Q}^{+}$ assigning weight $w_j$ to each subgraph $H_j$, the \textit{weight} of an independent $\mathcal{H}$-packing $\mathcal{H}' = \{H_i\}_{i\in I}$ in $G$ is $\sum_{i\in I}w_i$. Given a graph $G$, a finite family $\mathcal{H} = \{H_j\}_{j\in J}$ of connected non-null subgraphs of $G$, and a weight function $w\colon J \rightarrow \mathbb{Q}^{+}$ on the subgraphs in $\mathcal{H}$, the problem \textsc{Max Weight Independent Packing} asks to find an independent $\mathcal{H}$-packing in $G$ of maximum weight. In the special case when $\mathcal{F}$ is a \textit{fixed} finite family of connected non-null graphs and $\mathcal{H}$ is the set of all subgraphs of $G$ isomorphic to a member of $\mathcal{F}$, the problem is called \textsc{Max Weight Independent $\mathcal{F}$-Packing} and is a common generalization of several problems, among which: 
\textsc{Independent $\mathcal{F}$-Packing} \cite{CH06}, \textsc{Max Weight Independent Set} ($\mathcal{F} = \{K_1\}$), \textsc{Max Weight Induced Matching} ($\mathcal{F} = \{K_2\}$), \textsc{Dissociation Set} ($\mathcal{F} = \{K_1, K_2\}$ and the weight function assigns to each subgraph $H_j$ the weight $|V (H_j)|$) \cite{ODF11,Yan81}. 

\subsection{Packing subgraphs at distance at least $1$ in efficiently fractionally $\tin$-fragile classes}

Our PTAS relies on the following result.

\begin{theorem}[Dallard et el.\cite{DaMS22}]\label{MWIPA} Let $k$ and $h$ be two positive integers. Given a graph $G$ and a finite family $\mathcal{H} = \{H_j\}_{j\in J}$ of connected non-null subgraphs of $G$ such that $|V(H_j)| \leq h$ for every $j \in J$, \textsc{Max Weight Independent Packing} can be solved in time $O(|V(G)|^{h(k+1)} \cdot |V(T)|)$ if $G$ is given together with a tree decomposition $\mathcal{T} = (T, \{X_t\}_{t\in V(T)})$ with $\alpha(\mathcal{T}) \leq k$.
\end{theorem}

\begin{theorem}\label{PTASffA} Let $h \in \mathbb{N}$ and let $f \colon \mathbb{N} \rightarrow \mathbb{N}$ be a function. There exists an algorithm that, given 
\begin{itemize}
\item $r \in \mathbb{N}$, 
\item an $n$-vertex graph $G$ equipped with a $(1-1/r)$-general cover $\mathcal{C} = \{C_1,C_2,\ldots\}$ and, for each $i$, a tree decomposition $\mathcal{T}_i=(T_i,\{X_t\}_{t\in V(T_i)})$ of $G[C_i]$ with $\alpha(\mathcal{T}_i) \leq f(r)$, 
\item a finite family $\mathcal{H}=\{H_j\}_{j\in J}$ of connected non-null subgraphs of $G$ such that $|V(H_j)| \leq h$ for every $j \in J$, 
\item and a weight function $w\colon J \rightarrow \mathbb{Q}_{+}$ on the subgraphs in $\mathcal{H}$, 
\end{itemize}
returns in time $|\mathcal{C}|\cdot O(n^{h(f(r)+1)} \cdot t)$, where $t = \max_i |V(T_i)|$, an independent $\mathcal{H}$-packing in $G$ of weight at least a factor $(1-h/r)$ of the optimal. 
\end{theorem}

\begin{proof} For each $i \geq 1$, we proceed as follows. Using the algorithm from \Cref{MWIPA}, we simply compute a maximum-weight independent $\mathcal{H}$-packing $\mathcal{P}_i$ in $G[C_i]$ in time $O(n^{h(f(r)+1)} \cdot t)$. The total running time is then $|\mathcal{C}|\cdot O(n^{h(f(r)+1)} \cdot t)$. For a collection $\mathcal{A}$ of subgraphs of $G$, each isomorphic to a member of $\mathcal{H}$, and a subset $C \subseteq V(G)$, let $w(\mathcal{A}) = \sum_{A\in \mathcal{A}}w(A)$ and let $\mathcal{A} \cap C = \{A \in \mathcal{A} : A \subseteq C\}$. Observe that, given a subgraph $H$ of $G$, each vertex $v \in V(H)$ is not contained in at most $|\mathcal{C}|/r$ elements of the $(1-1/r)$-general cover $\mathcal{C}$. Hence, $V(H)$ is contained in at least $(1-|V(H)|/r)|\mathcal{C}|$ elements of $\mathcal{C}$. Let $\mathcal{P}=\{P_1,P_2,\ldots\}$ be an independent $\mathcal{H}$-packing in $G$ of maximum weight. Then, 

\begin{align*}
\sum_{C_i \in \mathcal{C}}w(\mathcal{P}\cap C_i)
& = \sum_{C_i \in \mathcal{C}} \sum_{P_j \in \mathcal{P}}  w(P_j) \mathbbm{1}_{\{P_j \subseteq C_i\}} \\
& =  \sum_{P_j \in \mathcal{P}}  w(P_j) \sum_{C_i \in \mathcal{C}}\mathbbm{1}_{\{P_j \subseteq C_i\}} \\
& \geq  \sum_{P_j \in \mathcal{P}}  w(P_j) (1-|V(P_j)|/r)|\mathcal{C}| \\
& \geq  \sum_{P_j \in \mathcal{P}}  w(P_j) (1-h/r)|\mathcal{C}| \\
& = |\mathcal{C}|(1-h/r)w(\mathcal{P}).
\end{align*}

By the pigeonhole principle, there exists $C_i \in \mathcal{C}$ such that $w(\mathcal{P}\cap C_i) \geq (1-h/r)w(\mathcal{P})$. We then return the maximum-weight independent $\mathcal{H}$-packing $\mathcal{P}_i$ in $G[C_i]$ computed above. Since $\mathcal{P} \cap C_i$ is an independent $\mathcal{H}$-packing in $G[C_i]$, we have that $w(\mathcal{P}_i) \geq w(\mathcal{P}\cap C_i) \geq (1-h/r)w(\mathcal{P})$.
\end{proof}

\Cref{PTASffA} immediately implies that \textsc{Max Weight Independent Packing} admits a $\mathsf{PTAS}$ in any efficiently fractionally $\tin$-fragile class. A special case is the following.

\begin{corollary} There exists an algorithm that, given $r \in \mathbb{N}$, a $c$-fat collection $\mathcal{O}$ of $n$ objects in $\mathbb{R}^d$ and its intersection graph $G$, returns in time $(f(r)/2-1)^d\cdot O(n^{(cf(r)^{2d}+2)})$, where $f(r) = 2\Big\lceil \frac{1}{1-\big(1-\frac{1}{r}\big)^{\frac{1}{d}}}\Big\rceil$, an independent set in $G$ of weight at least a factor $(1-1/r)$ of the optimal. 
\end{corollary}

\begin{proof} Given $r \in \mathbb{N}$, we use \Cref{treealphafatA} to compute in $O(n)$ time a $(1-1/r)$-general cover $\mathcal{C}$ of $G$ of size at most $(f(r)/2-1)^d$. Moreover, for each $C \in \mathcal{C}$, we compute in $O(n)$ time a tree decomposition $\mathcal{T} = (T,\{X_t\}_{t\in V(T)})$ of $G[C]$, with $|V(T)| \leq n+1$, such that $\alpha(\mathcal{T}) \leq cf(r)^{2d}$. We finally apply the algorithm from \Cref{PTASffA} (with $h=1$). The total running time is $(f(r)/2-1)^d\cdot O(n^{(cf(r)^{2d}+2)})$. 
\end{proof}


\subsection{Packing subgraphs at distance at least $d$ in graphs with bounded layered tree-independence number}\label{sec:distanceA}

\textsc{Max Weight Independent Packing} has a natural generalization. For a fixed positive integer $d$, given a graph $G$ and a finite family $\mathcal{H} = \{H_j\}_{j\in J}$ of connected non-null subgraphs of $G$, a \textit{distance-$d$ $\mathcal{H}$-packing} in $G$ is a subfamily $\mathcal{H}' = \{H_i\}_{i\in I}$ of subgraphs from $\mathcal{H}$ that are at pairwise distance at least $d$. If we are also given a weight function $w\colon J \rightarrow \mathbb{Q}_{+}$, \textsc{Max Weight Distance-$d$ Packing} is the problem of finding a distance-$d$ $\mathcal{H}$-packing in $G$ of maximum weight. The case $d = 2$ coincides with \textsc{Max Weight Independent Packing}.

\begin{theorem}\label{PTASltA} Let $h, \ell \in \mathbb{N}$. Let $d$ be an even positive integer. There exists an algorithm that, given 
\begin{itemize}
\item $r \in \mathbb{N}$, 
\item an $n$-vertex graph $G$ equipped with a tree decomposition $\mathcal{T} = (T,\{X_t\}_{t\in V(T)}\})$ and a layering $(V_1, V_2, \ldots)$ of $G$ such that, for each bag $X_t$ and layer $V_i$, $\alpha(G[X_t \cap V_i]) \leq \ell$, 
\item a finite family $\mathcal{H}=\{H_j\}_{j\in J}$ of connected non-null subgraphs of $G$ such that $|V(H_j)| \leq h$ for every $j \in J$, 
\item and a weight function $w\colon J \rightarrow \mathbb{Q}_{+}$, 
\end{itemize}
returns in time $r \cdot |V(T)| \cdot n^{O(r)}$ a distance-$d$ $\mathcal{H}$-packing in $G$ of weight at least a factor $(1-h/r)$ of the optimal. 
\end{theorem}

\begin{proof} Let $d = 2k$. As observed in \cite[Observation~3.9]{MR22}, for $I \subseteq J$, the subfamily $\mathcal{H}' = \{H_i\}_{i\in I}$ is a distance-$d$ $\mathcal{H}$-packing in $G$ if and only if $\mathcal{H}'$ is an independent $\mathcal{H}$-packing in the graph $G^{d-1}$. Therefore, using BFS, we first compute in $O(n^3)$ time the graph $G^{2k-1}$. Using the algorithm from \Cref{layeredtreealgoA}, we compute in $O(|V(T)| \cdot n^2)$ time a tree decomposition $\mathcal{T}' = (T,\{X'_t\}_{t\in V(T)})$ of $G^{2k-1}$ and a layering $(V'_1,\ldots,V'_{\lceil \frac{m}{2k-1} \rceil})$ of $G^{2k-1}$ such that, for each bag $X'_t$ and layer $V'_i$, $\alpha(G^{2k-1}[X'_t \cap V'_i]) \leq (4k-3)\ell$. Using the algorithm from \Cref{layeredtofragileA}, we compute in linear time a $(1 - 1/r)$-general cover $\mathcal{C}$ of $G^{2k-1}$ of size $r$ and, for each $C \in \mathcal{C}$, a tree decomposition of $G^{2k-1}[C]$ with tree-independence number at most $\ell (4k-3)(r-1)$. Finally, we apply the approximation algorithm from \Cref{PTASffA} to obtain, in time $r \cdot O(n^{h(\ell (4k-3)(r-1)+1)} \cdot |V(T)|)$, an independent $\mathcal{H}$-packing in $G^{2k-1}$ of weight at least a factor $(1-h/r)$ of the optimal.
\end{proof}

Combining \Cref{PTASltA} with \Cref{unitdisklayeredA}, we obtain the following:  

\begin{corollary} Let $d \in \mathbb{N}$ be even. \textsc{Max Weight Distance-$d$ Packing} admits a $\mathsf{PTAS}$ for unit disk graphs.
\end{corollary}

Observe that \Cref{PTASltA} cannot be extended to odd values of $d$, unless $\mathsf{P} = \mathsf{NP}$. Indeed, Eto et al. \cite{EGM14} showed that, for each $\varepsilon > 0$ and fixed odd $d \geq 3$, it is $\mathsf{NP}$-hard to approximate \textsc{Distance-$d$ Independent Set} to within a factor of $n^{1/2 - \varepsilon}$ for chordal graphs.

Since unit disk graphs have $O(\sqrt{n})$ tree-independence number (\Cref{unitdisklayeredA,sqrttreealphaA}) and since \textsc{Max Weight Distance-$d$ Packing} is solvable in time $n^{O(k)}$, where $k$ is the tree-independence number of the input graph \cite{MR22}, we immediately obtain a subexponential-time algorithm on unit disk graphs.  

\begin{lemma} For any fixed even $d \in \mathbb{N}$, \textsc{Max Weight Distance-$d$ Packing} can be solved in $2^{O(\sqrt{n}\log{n})}$ time on unit disk graphs.
\end{lemma}

A subexponential-time algorithm for \textsc{Independent Set} on unit disk graphs was first given in \cite{AF04} and later extended in \cite{BBK20} to intersection graphs of fat objects.


\subsection{Packing independent unit disks, unit-width rectangles and paths with bounded horizontal part on a grid}\label{sec:improvedA}

The next three results bound the tree-independence number of graphs whose geometric realizations are contained in an axis-aligned rectangle with bounded width and will be later applied to obtain PTASes for the classes mentioned in the title. It is worth noticing that, similarly to \Cref{sec:layeredlemmas}, all tree decompositions considered in this section are in fact path decompositions.

\begin{lemma}\label{pathA}
Let $G$ be an $n$-vertex graph together with a grid representation $\mathcal{R} = (\mathcal{G}, \mathcal{P},x)$ such that $\mathcal{G}$ contains at most $\ell$ columns, for some integer $\ell \geq 1$, and each path in $\mathcal{P}$ has number of bends constant. It is possible to compute in $O(n^2)$ time a tree decomposition $\mathcal{T} = (T, \{X_t\}_{t\in V(T)})$ of $G$ such that $|V(T)| = O(n^2)$ and:    
\begin{itemize}
\item $\alpha(\mathcal{T}) \leq \ell$, if $x = v$;
\item $\alpha(\mathcal{T}) \leq 3\ell - 1$, if $x = e$.  
\end{itemize}
\end{lemma}

\begin{proof}We assume without loss of generality that $G$ is connected. For each $u \in V(G)$ and $v \in N_G(u)$, we find an arbitrary grid-point that belongs to both $P_u$ and $P_v$ (since $uv \in E(G)$, such a grid-point exists). Let $p_{uv}$ be the projection of this grid-point onto the $y$-axis. Let $T$ be the vertical path whose nodes are the grid-points of the form $p_{uv}$. Clearly, $|V(T)| = O(n^2)$. Order the nodes of $T$ by decreasing $y$-coordinates of the corresponding grid-points. For each $t \in V(T)$, let $X_{t}' = \{u \in V(G) : p_{uv} = t \ \mbox{for some}\ v \in V(G)\}$. We finally update $\{X_t'\}_{t\in V(T)}$ as follows. For each $v \in V(G)$, let $t_{v}^{\min}$ and $t_{v}^{\max}$ be the smallest and largest node of $T$ such that the corresponding bag contains $v$, respectively. We then add $v$ to all the bags corresponding to nodes between $t_{v}^{\min}$ and $t_{v}^{\max}$. Let $\{X_t\}_{t\in V(T)}$ be the family of bags thus obtained and let $\mathcal{T} = (T, \{X_t\}_{t\in V(T)})$. As in the proof of \Cref{layeredVPGA}, $\mathcal{T}$ is a tree decomposition of $G$ that can be constructed in $O(n^2)$ time. 

We now show that, for each $t \in V(T)$, $\alpha(G[X_t]) \leq \ell$ if $x = v$, whereas $\alpha(G[X_t]) \leq 3\ell-1$ if $x = e$. This would imply that $\alpha(\mathcal{T}) \leq \ell$ and $\alpha(\mathcal{T}) \leq 3\ell - 1$, respectively, thus concluding the proof. Let $I$ be an independent set of $G[X_t]$. Observe that, by construction, for each $v \in I \subseteq X_t$, the path $P_v$ contains a grid-point on the row of $\mathcal{G}$ indexed by $t$. Suppose first that $x = v$ i.e., $G$ is a VPG graph. Then, each grid-point on the row indexed by $t$ belongs to at most one $P_v$ with $v \in I$. Since $\mathcal{G}$ has at most $\ell$ columns, $|I| \leq \ell$. Suppose finally that $x = e$ i.e., $G$ is an EPG graph. Then, since there are $3\ell - 1$ grid-edges containing a grid-point on the row indexed by $t$, we have that $|I| \leq 3\ell - 1$, or else two paths $P_u$ and $P_v$ with $u, v \in I$ share a grid-edge.                   
\end{proof}

The proofs of the following two results are similar to that of \Cref{pathA}.

\begin{lemma}\label{rectangleA}
Let $G$ be the intersection graph of a family of $n$ rectangles together with a grid representation $(\mathcal{G}, \mathcal{R})$ such that $\mathcal{G}$ contains at most $\ell$ columns, for some integer $\ell \geq 1$. It is possible to compute in $O(n^2)$ time a tree decomposition $\mathcal{T} = (T, \{X_t\}_{t\in V(T)})$ of $G$ such that $|V(T)| = O(n^2)$ and $\alpha(\mathcal{T}) \leq \lfloor\frac{\ell}{2}\rfloor$. 
\end{lemma}

\begin{proof} We assume without loss of generality that $G$ is connected. For each $u \in V(G)$ and $v \in N_G(u)$, we find an arbitrary grid-point that belongs to both rectangles $R_u$ and $R_v$ (since $uv \in E(G)$, such a grid-point exists). Let $p_{uv}$ be the projection of this grid-point onto the $y$-axis and let $T$ be the vertical path whose nodes are the grid-points of the form $p_{uv}$. We then build $\mathcal{T} = (T, \{X_t\}_{t\in V(T)})$ precisely as in the proof of \Cref{pathA}. It is again easy to see that $\mathcal{T}$ is a tree decomposition of $G$ that can be constructed in $O(n^2)$ time.  

We now show that, for each $t \in V(T)$, $\alpha(G[X_t]) \leq \lfloor\frac{\ell}{2}\rfloor$. This would imply that $\alpha(\mathcal{T}) \leq \lfloor\frac{\ell}{2}\rfloor$, thus concluding the proof. Let $I$ be an independent set of $G[X_t]$. Observe that, by construction, for each $v \in I \subseteq X_t$, the rectangle $R_v$ contains at least two grid-points on the row of $\mathcal{G}$ indexed by $t$, and each grid-point on the row indexed by $t$ belongs to at most one $R_v$ with $v \in I$. Since $\mathcal{G}$ has at most $\ell$ columns, $I \leq \lfloor\frac{\ell}{2}\rfloor$. 
\end{proof}

\begin{lemma}\label{diskA}
Let $G$ be the intersection graph of a family $\mathcal{D}$ of $n$ unit disks of common radius $c \geq 1$ such that its geometric realization is contained in an axis-aligned rectangle with integral vertices and width at most $\ell-1$, for some integer $\ell \geq 1$. It is possible to compute in $O(n)$ time a tree decomposition $\mathcal{T} = (T, \{X_t\}_{t\in V(T)})$ of $G$ such that $|V(T)| = O(n)$ and $\alpha(\mathcal{T}) \leq 2\lceil\frac{\ell}{c}\rceil$. 
\end{lemma}

\begin{proof}
We assume without loss of generality that $G$ is connected and that the collection of disks is contained in the positive quadrant. For each integer $j > 0$, let $R_j = \{(x,y) \in \mathbb{R}^2: c(j-1) \leq y \leq cj\}$ and $C_j = \{(x,y) \in \mathbb{R}^2: c(j-1) \leq x \leq cj\}$ be the \textit{$j$-th horizontal strip} and the \textit{$j$-th vertical strip}, respectively. Note that for every $v \in V(G)$, the disk $D_v$ intersects at most three horizontal strips and at most three vertical strips. In particular, $m = |\{j \in \mathbb{N}: R_j \cap \{D_v:v \in V(G)\} \neq \varnothing\}| \leq 3n$. Without loss of generality, we may assume that the first $m$ horizontal strips are non-empty, that is, $R_j \cap \{D_v:v \in V(G)\} \neq \varnothing$ for every $j \in [m]$. Now by assumption $|\{j \in \mathbb{N}: C_j \cap \{D_v:v \in V(G)\} \neq \varnothing\}| \leq \lceil \frac{\ell}{c} \rceil$ and we may similarly assume that the non-empty vertical strips are the first $\lceil \frac{\ell}{c}\rceil$ ones. We now construct a tree decomposition of $G$ as follows. Let $T$ be a path on $m$ vertices $t_1,\ldots,t_m$ and let $X_{t_j} = \{v \in V(G): D_v \cap R_j \neq \varnothing\}$ for every $j \in [m]$. Clearly, for each $v \in V(G)$, there exists $j \in [m]$ such that $v \in X_{t_j}$. Consider now $uv \in E(G)$. Then, there exists a point $(x,y) \in  \mathbb{R}^2$ contained in both $D_v$ and $D_u$ and so $\{u,v\} \subseteq X_{t_{\lceil \frac{y}{c} \rceil}}$. Finally, for every $v \in V(G)$, the horizontal strips intersecting $D_v$ are consecutive and so the set of nodes of $T$ whose bag contains $v$ induces a subpath. Therefore, $\mathcal{T} = (T,\{X_{t_j}\}_{1 \leq j \leq m})$ is a tree decomposition of $G$ and it is easy to see that this tree decomposition can be computed in linear time. 

We now claim that $\alpha(\mathcal{T}) \leq 2 \lceil \frac{\ell}{c}\rceil$. Indeed, consider $j \in [m]$ and let $I \subseteq X_{t_j}$ be an independent set of $G[X_{t_j}]$. We show that, for each $i \in [\lceil \frac{\ell}{c} \rceil]$, there are at most two vertices $v \in I$ such that $C_i$ contains the center of $D_v$. To this end, for every $p,q \in \mathbb{N}$, denote by $B(p,q) = \{(x,y) \in \mathbb{R}^2: c(p-1) \leq x \leq cp \text{ and } c(q-1) \leq y \leq cq\}$ the box at the intersection of the $p$-th vertical strip and $q$-th horizontal strip. Observe that, for any $p,q \in \mathbb{N}$, $B(p,q)$ contains the center of at most one disk in $\{D_v: v \in I\}$, as any two points in $B(p,q)$ are at distance at most $\sqrt{2} c$. Consider now $i \in \lceil \frac{\ell}{c} \rceil$. Since for any disk $D \in \{D_v: v \in I\}$ whose center $c$ is contained in $C_i$, there exists $p \in \{j-1,j,j+1\}$ such that $c \in B(i,p)$, it follows that $C_i$ contains the center of at most three disks in $\{D_v: v \in I\}$. Now suppose, to the contrary, that $C_i$ contains the center of three disks in $\{D_v: v \in I\}$, say $c_{j-1} = (x_{j-1},y_{j-1}), c_j = (x_j,y_j)$ and $c_{j+1}=(x_{j+1},y_{j+1})$, where $c_p \in B(i,p)$ for each $p \in \{j-1,j,j+1\}$. Then, either $y_j - y_{j-1} \leq 3c/2$ or $y_{j+1} - y_j \leq 3c/2$ and, assuming without loss of generality that the latter holds, we conclude that $(x_{j+1} - x_j)^2 + (y_{j+1} - y_j)^2 \leq c^2(1 + 9/4) < 4c^2$, a contradiction to the fact that $I$ is an independent set. Therefore $|I| \leq 2 \lceil \frac{\ell}{c}\rceil$ and so $\alpha(\mathcal{T}) \leq 2 \lceil \frac{\ell}{c}\rceil$ as claimed.
\end{proof}

\begin{theorem}\label{indepPTASA} Let $c \geq 1$ be an integer constant. \textsc{Max Weight Independent Set} admits a PTAS when restricted to $n$-vertex graphs with a grid representation $\mathcal{R} = (\mathcal{G}, \mathcal{P},x)$ such that: 
\begin{enumerate}
\item each path in $\mathcal{P}$ has number of bends constant;
\item\label{3rdA} the horizontal part of each path in $\mathcal{P}$ has length at most $c$.
\end{enumerate}
If $x = v$, the running time is $O(c\lceil\frac{1}{\varepsilon}\rceil \cdot n^{\lceil\frac{1}{\varepsilon}\rceil c + 4})$. If $x = e$, the running time is $O(c\lceil\frac{1}{\varepsilon}\rceil \cdot n^{3(\lceil\frac{1}{\varepsilon}\rceil c + 1)})$. 
\end{theorem}

\begin{proof}
Let $G$ be a graph with a grid representation $\mathcal{R} = (\mathcal{G}, \mathcal{P},x)$ satisfying the three conditions above. Without loss of generality, we may assume that all the paths in $\mathcal{P}$ contain only grid-points with non-negative coordinates. Moreover, we may assume that $G$ is connected. Therefore, no column in $\mathcal{G}$ is unused and so $\mathcal{G}$ has at most $(c+1)n$ columns. Further note that since any path $P \in \mathcal{P}$ has number of bends constant, we can compute the horizontal part $h(P)$ of $P$ in $O(1)$ time. Given $0 < \varepsilon < 1$, we fix $k = \lceil 1/\varepsilon \rceil$. 

For any $i \geq 0$, we denote by $X_i$ the set of vertices whose corresponding path contains a grid-edge $[(i,j),(i+1,j)]$ for some $j \geq 0$ (here and in the following $[(i,j),(i+1,j)]$ denotes the grid-edge with endpoints $(i,j)$ and $(i+1,j)$). Note that we can compute the at most $(c+1)n -1$ non-empty sets $X_{i}'s$ in $O(n)$ time. In view of applying a shifting technique, we now partition $G$ into slices via the following. For any $d \in \{0, \ldots, kc-1\}$, let $V_d = \bigcup_{\ell \in \mathbb{N}_0} X_{d + \ell kc}$ be the set of vertices whose corresponding path contains a grid-edge $[(d +\ell kc,j),(d+\ell kc +1,j)]$ for some $\ell,j \in \mathbb{N}$. We claim that, for any $d \in \{0,\ldots,kc-1\}$, $G - V_d$ is disconnected. Indeed, after deleting $V_d$, no vertex whose horizontal part is contained in the interval $[0, d + \ell kc]$ can be adjacent to a vertex whose horizontal part is contained in the interval $[d + \ell kc+1, (c+1)n]$. Similarly, every component of $G - V_d$ admits a grid representation in which the number of columns is bounded by $kc$. By \Cref{pathA} and \Cref{MWIPA}, for each component of $G - V_d$, we compute a maximum-weight independent set in $O(n^{kc+3})$ time, if $x = v$, or in $O(n^{3kc+2})$ time, if $x = e$. The union $U_d$ of these independent sets over the components of $G - V_d$ is then an independent set of $G$ and, after repeating the procedure above for each $d \in \{0,\ldots, kc-1\}$, we return the maximum-weight set $U$ among the $U_d$'s. The total running time is then $O(kc \cdot n^{kc+4})$, if $x = v$, or $O(kc \cdot n^{3kc+3})$, if $x = e$.  
 
It remains to show that $w(U) \geq (1 - \varepsilon)w(\mathsf{OPT})$, where $\mathsf{OPT}$ denotes an optimal solution of \textsc{Max Weight Independent Set} with instance $G$. Note that, for any $d \in \{0,\ldots , kc-1\}$, $\mathsf{OPT} \cap V_d$ is the set of vertices in $\mathsf{OPT}$ whose corresponding path contains a grid-edge $[(d +\ell kc,j),(d+\ell kc +1,j)]$ for some $\ell,j \in \mathbb{N}_0$. Since the horizontal part of each path has length at most $c$, we have that every vertex in $\mathsf{OPT}$ belongs to at most $c$ distinct $V_d$'s. Therefore, denoting by  $d_0$ the index attaining $\min_{d \in \{0, \ldots, kc-1\}} w(\mathsf{OPT} \cap V_d)$, we have 
\[
kc\cdot w(\mathsf{OPT} \cap V_{d_0}) \leq \sum_{d=0}^{kc-1} w(\mathsf{OPT} \cap V_d) \leq c\cdot w(\mathsf{OPT})
\] and so 
\[
w(\mathsf{OPT}) = w(\mathsf{OPT}\setminus V_{d_0}) + w(\mathsf{OPT} \cap V_{d_0}) \leq w(U) + \varepsilon \cdot w(\mathsf{OPT}),
\]
thus concluding the proof.
\end{proof}

\begin{theorem}\label{rectanglesPTASA} \textsc{Max Weight Independent Set} admits a PTAS when restricted to:
\begin{itemize}
\item Intersection graphs of a family of $n$ width-$c$ rectangles together with a grid representation $(\mathcal{G}, \mathcal{R})$. The running time is $O(c\lceil \frac{1}{\varepsilon} \rceil \cdot n^{\lceil \frac{1}{\varepsilon} \rceil\cdot\frac{c}{2}+4})$.
\item Intersection graphs of a family $\mathcal{D}$ of $n$ unit disks of common radius $c \geq 1$. The running time is $O(c\lceil \frac{2}{\varepsilon} \rceil \cdot n^{2\lceil \frac{2}{\varepsilon} \rceil+3})$. 
\end{itemize}
\end{theorem}

\begin{proof} Since the PTASes are similar, we introduce some common notation. Let $G$ be the intersection graph of the family $\mathcal{O}$, where $\mathcal{O}$ is either $\mathcal{R}$ or $\mathcal{D}$. Without loss of generality, we may assume that all objects in $\mathcal{O}$ are contained in the positive quadrant. Moreover, we may assume that $G$ is connected. Therefore, $\mathcal{O}$ is contained in a grid $\mathcal{G}$ with $O(n)$ columns. For each $O \in \mathcal{O}$, we compute the horizontal part $h(O)$ of $O$ (i.e., the projection of $O$ onto the $x$-axis) in $O(1)$ time. Given $0 < \varepsilon < 1$, we fix $k = \lceil 1/\varepsilon \rceil$, in the case of rectangles, and $k = \lceil 2/\varepsilon \rceil$, in the case of disks. Let $X_i$ be the set of vertices whose corresponding objects have horizontal part intersecting the grid-edge $[(i,0),(i+1,0)]$. We can compute the $X_i$'s in $O(n)$ time. We now partition $G$ into slices as follows. For any $d \in \{0, \ldots, kc-1\}$, let $V_d = \bigcup_{\ell \in \mathbb{N}_0} X_{d + \ell kc}$. As in the proof of \Cref{indepPTASA}, it is easy to see that, for any $d \in \{0,\ldots,kc-1\}$, $G - V_d$ is disconnected and that every component of $G - V_d$ admits a geometric realization which is contained in an axis-aligned rectangle with integral vertices and width at most $kc - 1$. If the family consists of rectangles, by \Cref{rectangleA} and \Cref{MWIPA}, for each component of $G - V_d$, we compute a maximum-weight independent set in $O(n^{\frac{kc}{2} + 3})$ time. If the family consists of disks, by \Cref{diskA} and \Cref{MWIPA}, for each component of $G - V_d$, we compute a maximum-weight independent set in $O(n^{2k + 2})$ time. In either case, the union $U_d$ of these independent sets over the components of $G - V_d$ is an independent set of $G$ and, after repeating the procedure above for each $d \in \{0,\ldots, kc-1\}$, we return the maximum-weight set $U$ among the $U_d$'s. The total running time is then $O(kc \cdot n^{\frac{kc}{2}+4})$, in the case of rectangles, and $O(kc \cdot n^{2k + 3})$, in the case of disks. Similarly to \Cref{indepPTASA}, it is easy to see that $w(U) \geq (1 - \varepsilon)w(\mathsf{OPT})$, where $\mathsf{OPT}$ denotes an optimal solution of \textsc{Max Weight Independent Set} with instance $G$. 
\end{proof}

\bibliography{references}

\end{document}